\documentclass[11pt]{article}
\usepackage[english]{babel}
\usepackage[toc,page]{appendix}
\usepackage{bbm}
\usepackage{fullpage}
\usepackage[top=1in, bottom=1.25in, left=0.9in, right=0.9in]{geometry}
\usepackage{etaremune}
\usepackage{enumitem}

\usepackage{tikz}
\usetikzlibrary{positioning}
\usepackage{witharrows}
\usepackage{amssymb,amsmath,amsthm}

\usepackage{xcolor}
\definecolor{myblue}{RGB}{0,0,255}
\definecolor{myyellow}{RGB}{255,230,80}

\usepackage{hyperref}

\DeclareSymbolFont{rsfs}{U}{rsfs}{m}{n}
\DeclareSymbolFontAlphabet{\mathscrsfs}{rsfs}
\usepackage{mathrsfs}

\usepackage{url}

\hypersetup{
  colorlinks   = true, 
  urlcolor     = blue, 
  linkcolor    = blue, 
  citecolor   = red 
}

\usepackage{caption}

\newtheorem{theorem}{Theorem}[section]
\newtheorem{lemma}[theorem]{Lemma}

\newtheorem{proposition}[theorem]{Proposition}
\newtheorem{corollary}[theorem]{Corollary}

\theoremstyle{definition}
\newtheorem{definition}{Definition}
\newtheorem{remark}[theorem]{Remark}

\newtheorem{problem}{Open problem}

\numberwithin{equation}{section}

\usepackage{times}

\newcommand{\dH}{d_{\mathrm{H}}}
\newcommand{\Ball}{\mathsf{B}}
\newcommand{\dist}{\mathrm{dist}}
\newcommand{\proj}{\Pi}

\newcommand{\bea}{\begin{eqnarray}}
\newcommand{\eea}{\end{eqnarray}}
\newcommand{\<}{\langle}
\renewcommand{\>}{\rangle}

\newcommand{\wt}{\widetilde}

\newcommand\eg{{\text{\eg~}}}

\def\ind{{\mathbbm 1}}

\def\btau{{\boldsymbol{\tau}}}
\def\bsigma{{\boldsymbol{\sigma}}}

\def\bSigma{{\boldsymbol{\Sigma}}}
\def\bSig{{\boldsymbol{\Sigma}}}

\def\bg{{\boldsymbol{g}}}

\def\bx{{\boldsymbol{x}}}

\def\cG{{\mathcal G}}

\def\bsig{{\boldsymbol {\sigma}}}

\def\by{{\boldsymbol y}}

\def\R{{\mathbb R}}

\def\bx{{\boldsymbol{x}}}

\def\bH{\boldsymbol{H}}

\def\bX{\boldsymbol{X}}

\def\<{\langle}
\def\>{\rangle}

\def\Ball{{\sf B}}

\def\cN{{\cal N}}

\def\by{{\boldsymbol{y}}}

\def\b0{{\boldsymbol{0}}}

\def\bG{{\boldsymbol G}}

\def\cS{{\mathcal S}}

\def\star{*}

\renewcommand{\b}{\mathbf{b}}

\def\la{\langle}
\def\ra{\rangle}

\def\bbE{{\mathbb{E}}}

\def\bbN{{\mathbb{N}}}
\def\bbP{{\mathbb{P}}}
\def\bbR{{\mathbb{R}}}

\def\cN{{\mathcal{N}}}

\def\bg{{\mathbf{g}}}

\def\star{*}

\title{
Stable algorithms cannot reliably find isolated perceptron solutions
}
\author{Shuyang Gong\footnote{School of Mathematical Sciences, Peking University. Email: \textit{gongsyprob@gmail.com}.}
\and Brice Huang\footnote{Department of Statistics, Stanford University. Email: \textit{bmhuang@stanford.edu}.}
\and Shuangping Li\footnote{Department of Statistics and Data Science, Yale University. Email: \textit{shuangping.li@yale.edu}.} 
\and Mark Sellke\footnote{Department of Statistics, Harvard University. Email: \textit{msellke@fas.harvard.edu}.}
}
\date{}

\begin{document}

	\maketitle

	\begin{abstract}

We study the binary perceptron, a random constraint satisfaction problem that asks to find a Boolean vector in the intersection of independently chosen random halfspaces. A striking feature of this model is that at every positive constraint density, it is expected that a $1-o_N(1)$ fraction of solutions are \emph{strongly isolated}, i.e. separated from all others by Hamming distance $\Omega(N)$. At the same time, efficient algorithms are known to find solutions at certain positive constraint densities. This raises a natural question: can any isolated solution be algorithmically visible?

We answer this in the negative: no algorithm whose output is stable under a tiny Gaussian resampling of the disorder can \emph{reliably} locate isolated solutions. We show that any stable algorithm has success probability at most $\frac{3\sqrt{17}-9}{4}+o_N(1)\leq 0.84233$. Furthermore, every stable algorithm that finds a solution with probability $1-o_N(1)$ finds an isolated solution with probability $o_N(1)$. The class of stable algorithms we consider includes degree-$D$ polynomials up to $D\leq o(N/\log N)$; under the low-degree heuristic \cite{hopkins2018statistical}, this suggests that locating strongly isolated solutions requires running time $\exp(\widetilde{\Theta}(N))$.

Our proof does not use the overlap gap property. Instead, we show via Pitt's correlation inequality that after a random perturbation of the disorder, the number of solutions located close to a pre-existing isolated solution cannot concentrate at $1$.

\end{abstract}

\section{Introduction}

The perceptron model is a simple one-layer neural network model that has attracted significant recent attention at the interface of statistical physics, probability theory, learning theory, and theoretical computer science. It was introduced by Cover \cite{cover1965geometrical} and later studied extensively, notably by Gardner and Derrida \cite{gardner1988optimal} and by Krauth and M\'ezard \cite{krauth1989storage}. Let $\bg^a \in \mathbb{R}^N$, for $1 \le a \le M$, be i.i.d.\ random patterns. We say that these patterns are \emph{storable} if there exists a vector $\bsig\neq \mathbf 0$ such that $\langle \bg^a,\bsig\rangle \ge 0$ for all $1\le a\le M$. The perceptron may thus be viewed as a random constraint satisfaction problem (CSP), in which one seeks a vector satisfying $M$ random half-space constraints.

There are two widely studied versions of the perceptron model. In the \emph{spherical perceptron}, the vector $\bsig$ is constrained to lie on the unit sphere $\mathbb S^{N-1}$. In the \emph{binary} or \emph{Ising perceptron}, one instead requires $\bsig\in \bSig_N := \{-1,+1\}^N$. For background on the spherical perceptron, we refer the reader to \cite{gardner1988space,shcherbina2003rigorous,talagrand2011mean2,stojnic2013negative,stojnic2013another,el2022algorithmic}. In this paper, we focus on the binary perceptron.

A convenient way to formulate binary perceptron models is through an activation function $U:\mathbb{R}\to\{0,1\}$, where a pattern $\bg^a$ is said to be stored by $\bsig$ if $U(\langle \bg^a,\bsig\rangle/\sqrt{N})=1$. Two choices of particular interest are $U(t)=\mathbf{1}_{\{t\ge \kappa\}}$ and $U(t)=\mathbf{1}_{\{|t|\le \kappa\}}$. These correspond to the \emph{asymmetric binary perceptron} (ABP) and the \emph{symmetric binary perceptron} (SBP), respectively. We describe these models in more detail below.

\subsection{Perceptron models}
We begin with the asymmetric binary perceptron. Let the disorder matrix $\bG=(\bg^1,\dots,\bg^M)\in(\bbR^N)^M$ have i.i.d.\ standard Gaussian entries $\cN(0,1)$, 
and fix a margin $\kappa\in\bbR$. A \emph{solution} to the perceptron problem is a configuration $\bsig\in\bSig_N$ such that $\langle \bg^a,\bsig\rangle\ge \kappa\sqrt N$ for all $1\le a\le M$. Accordingly, the set of all solutions is
\[
\cS(\bG,\kappa)\;=\;\bigcap_{a=1}^M\Bigl\{\bsig\in\bSig_N:\ \langle \bg^a,\bsig\rangle\ge \kappa\sqrt N\Bigr\}.
\]
Throughout, we work in the asymptotic regime $M,N\to\infty$ with $M/N\to\alpha\in(0,\infty)$. The parameter $\alpha$ is referred to as the \emph{constraint density}.

The symmetric binary perceptron model was introduced by Aubin, Perkins, and Zdeborov\'a in \cite{aubin2019storage}. It is defined by replacing the one-sided constraint with the two-sided slab
\begin{equation}
    \mathcal{S}_{\mathtt{SBP}}(\bG,\kappa):=\bigcap_{a=1}^M \Big\{ \bsig\in\bSig_N:\,|\langle\bg^a,\bsig\rangle|\leq \kappa\sqrt{N} \Big\}\,.
\end{equation}
The solution set $\mathcal{S}_{\mathtt{SBP}}(\bG,\kappa)$ is symmetric. This makes the symmetric binary perceptron more amenable to analysis while still preserving many of the key qualitative features of the asymmetric binary perceptron (see the discussion of strong freezing below). 
The model is also closely connected to combinatorial discrepancy theory \cite{spencer1985six,matousek1999geometric}: indeed, it can be viewed as a discrepancy problem in the random proportional regime.

\paragraph{Satisfiability threshold.}
An important question for the perceptron model is the satisfiability threshold, namely the critical constraint density $\alpha_{\mathrm{SAT}}(\kappa)$ at which $\cS(\bG,\kappa)$ transitions from being nonempty with high probability to empty with high probability. For the ABP, following the statistical physics prediction of Krauth and M\'ezard \cite{krauth1989storage}, substantial rigorous progress has been made on this question, including upper bounds in \cite{kim1998covering,altschuler2024note}, a sharp lower bound in \cite{ding2019capacity}, a matching sharp upper bound in \cite{huang2024capacity} (both subject to explicit numerical conditions), sharp-threshold results in \cite{talagrand2011mean2,xu2021sharp,nakajima2023sharp}, and a formula for the free energy at low constraint density \cite{bolthausen2022gardner}. For the SBP, the satisfiability threshold was established in \cite{aubin2019storage,PerkinsXu21,AbbeLiSly2021}. In addition, the model admits a more refined analysis, including contiguity results \cite{AbbeLiSly2021} and detailed studies of the critical window~\cite{altschuler2023critical,sah2023distribution}.

\paragraph{Algorithms.}
Once satisfiability is understood, the next question is algorithmic: whenever solutions exist, is there a polynomial-time algorithm that finds at least one solution? In the statistical physics perspective, this algorithmic question is closely tied to the geometry of the solution space. Roughly speaking, if the set of solutions is poorly connected inside the hypercube, so that typical solutions are far from one another, then one expects algorithms or at least natural local dynamics to have difficulty finding solutions or moving between them. This connection between geometry and algorithms is especially intriguing for perceptron models, and will be directly relevant to the questions studied in this paper.

For the asymmetric binary perceptron, several efficient algorithms are known in certain regimes. Kim and Roche~\cite{kim1998covering} introduced a multistage majority algorithm that succeeds with high probability when $\kappa=0$ and $\alpha<0.005$. More recently, Li, Schramm, and Zhou~\cite{li2025discrepancy} improved this guarantee to $\alpha\le 0.1$ by developing an algorithm based on discrepancy theory, and also obtained sharp algorithmic results in the large-$|\kappa|$ regime. In addition, a number of heuristic message-passing algorithms, particularly those based on belief propagation, have been studied in the physics and machine learning literature; see e.g.\ \cite{braunstein2006learning,baldassi2007efficient,baldassi2009generalization,baldassi2015max}. 
For the symmetric binary perceptron, its close connection to discrepancy minimization leads to efficient algorithms in certain parameter regimes, as well as hardness results in others; see, for example,~\cite{bansal2010constructive,chandrasekaran2014integer,lovett2015constructive,bansal2020line,potukuchi2019spectral,alweiss2021discrepancy,gamarnik2022algorithms,gamarnik2023geometric,MR4928600}.

\paragraph{Strong freezing.}
To further understand the algorithmic question, one is naturally led to the geometry of the solution space. In sparse random CSPs, two fundamental geometric properties of the solution space are \emph{clustering} and \emph{freezing}. Clustering refers to the shattering of the solution space into exponentially many connected components separated by linear distance, a picture formulated in \cite{krzakala2007gibbs} and proved in \cite{achlioptas2008algorithmic}. Freezing means that within a typical cluster, a linear fraction of variables take the same value across all solutions. For example, in random graph coloring, the freezing threshold was established in \cite{molloy2012freezing}; see also \cite{zdeborova2008constraint,zdeborova2008locked,zdeborova2011quiet} for related results. In many such models, including random $k$-SAT, $k$-NAESAT, and random graph coloring, both clustering and freezing occur well before the satisfiability threshold.

Perceptron models are unusual in that clustering and freezing are conjectured to occur at all positive constraint densities $\alpha>0$. Furthermore, it is expected that clustering occurs in a strong sense: \cite{krauth1989storage} conjectured that in the ABP, typical solutions lie in clusters of vanishing entropy density (i.e., containing only $2^{o(N)}$ solutions), with distinct clusters separated by linear Hamming distance. 
This prediction was strengthened in \cite{HuangWongKabashima2013,HuangKabashima2014}, which conjectured that typical solutions are genuinely isolated up to linear Hamming distance.
These predictions were later extended to the SBP in \cite{aubin2019storage,baldassi2020clustering}. 

This picture was rigorously confirmed for the SBP in \cite{PerkinsXu21,AbbeLiSly2021}, where it is shown that, with high probability, all but an $o_N(1)$-fraction of solutions are separated from every other solution by a linear Hamming distance, for every $\alpha>0$. The same behavior is widely believed to hold for the ABP.

\paragraph{Why algorithms may still succeed.}
At first glance, the strong-freezing picture seems difficult to reconcile with the existence of efficient algorithms for finding perceptron solutions at positive density.
Indeed, broadly in disordered search and optimization problems, a long-standing belief from the statistical physics literature is that the geometry of the solution space is closely linked to the computational tractability of the problem.
For example, \cite{krzakala2007gibbs} conjectures that a broad class of Markov chain algorithms sample from the solution space up to the clustering transition and no further, while an influential conjecture of \cite{achlioptas2008algorithmic} suggests that the clustering transition is the threshold beyond which any polynomial-time algorithm cannot find a solution.

However, for sparse random CSPs with bounded typical degree, it has been known since \cite{achlioptas2002almost,zdeborova2007phase} that search algorithms can succeed beyond the clustering transition. This phenomenon is even more striking in random perceptron models: even though strong freezing is expected at all constraint densities $\alpha > 0$ for the ABP (and has been proved for the SBP), many efficient algorithms are nevertheless known to find solutions at some positive $\alpha$. This creates an apparent violation of the heuristic that clustering of the solution space implies hardness for algorithms.

A possible resolution was proposed in \cite{baldassi2015subdominant,baldassi2020clustering}: although a $(1-o_N(1))$-fraction of solutions are strongly frozen, an exponentially small fraction may belong to clusters of exponential size, and efficient algorithms may succeed by finding solutions in these rare but well-connected clusters. To turn this picture into a rigorous explanation, one must establish two complementary facts. First, such rare but well-connected clusters indeed exist. Second, strongly isolated solutions must themselves be hard to find algorithmically. The first direction has seen recent progress; the second is addressed here. Supporting the first direction, Abbe, Li, and Sly \cite{abbe2022binary} proved that, for small $\alpha$, with high probability there exists a cluster of nearly maximal diameter; moreover, such a cluster can be found efficiently by a suitable multiscale majority algorithm in the spirit of \cite{kim1998covering}. As further evidence for this picture, the \emph{branching overlap gap property} (see Subsection~\ref{ss:ogp} below) of Huang and Sellke \cite{huang2025tight} rigorously links algorithmic tractability to the presence of a certain type of dense solution cluster, namely a densely-branching ultrametric tree. Complementary to these results, several recent works further investigate the structure of such atypical solution sets; see, for example, \cite{barbier2024atypical,barbier2025finding,barbier2026low,stojnic2025rare}.

\subsection{Main results: algorithmic hardness of locating isolated solutions}
The rare but well-connected cluster picture described above leaves one last piece of the puzzle: are isolated solutions hard to find algorithmically? In other words, although efficient algorithms may succeed by finding such atypical clusters, do they necessarily avoid the strongly isolated solutions that make up a $1-o_N(1)$ fraction of the solution space? This is the question we address in this paper.

This question arises implicitly in~\cite{PerkinsXu21} from the strong-freezing picture, and was explicitly raised as an open problem in~\cite{el2025hardness}. In \cite{el2025hardness}, it is proved that stable algorithms cannot produce an approximate \emph{sample} from the uniform measure over the solution space of the SBP model at any density $\alpha>0$. But this does not a priori exclude the possibility that some special isolated solutions are algorithmically visible.

In this paper, we answer this question affirmatively. We formulate the obstruction in terms of stability under a small Gaussian resampling of the disorder, which includes polynomial algorithms of degree $o(N/\log N)$; see Corollary~\ref{cor:ld-stable}. Our main results are as follows.
\begin{itemize}
    \item Stable algorithms locate\footnote{Here, ``locating'' a solution means that the algorithm outputs some $\bx\in\mathbb R^N$ that is close to an isolated solution. See Definition~\ref{def:locating}.} a strongly isolated solution with probability at most $\approx 0.84233$ (Theorem~\ref{thm:stable-v1}).
    \item For any stable algorithm, including the algorithm of \cite{kim1998covering}, that locates a solution with probability $1-o_N(1)$, the probability that it locates a strongly isolated solution is $o_N(1)$ (Theorem~\ref{thm:stable-v2}).
\end{itemize}
Similar arguments yield finer-scale $k$-isolation statements and extensions to perceptrons defined by a finite union of intervals, including the SBP; though as we explain, the implications here are less clear.
Our main result shows hardness in the sense that algorithms succeed with probability at most $\approx 0.84233$. Although one would ideally like to show that algorithms succeed with probability $o_N(1)$, we remark that hardness results with a non-vanishing success probability are common in the literature, and are still meaningful evidence of computational intractability. Indeed, in random search and optimization problems, it is expected that in the computationally easy phase, polynomial-time algorithms can succeed with probability very close to $1$. Thus a result that algorithms cannot succeed with probability larger than $1-\Omega(1)$ still suggests computational hardness, and \cite{GJW20,gamarnik2024hardness,wein2022optimal,bresler2022algorithmic,gamarnik2022algorithms,li2025discrepancy} show similar hardness results in related problems. We leave the problem of strengthening the success probability to $o_N(1)$ as an important direction for future work.

We now give the formal definitions and statements of our main results. For $\bsig,\btau\in \bSig_N$, we define the Hamming distance by
\[
\dH(\bsigma,\btau) := \frac12\sum_{i=1}^N |\bsigma_i-\btau_i|\in\{0,1,\dots,N\}.
\]
For $r\ge 0$, we write the Hamming ball $\Ball_{r,\mathrm{H}}(\bsigma):=\{\btau\in\bSigma_N : \dH(\bsigma,\btau)\le r\}$. For each $x\in \mathbb R^N$ and $r\geq 0$, we also define the $\ell_2$ ball $\mathsf{B}_{r,2}(x):=\{y\in \mathbb R^N:\|y-x\|_2\leq r\}$. For a set $T\subseteq\bSigma_N$ and point $\bx\in\R^N$, write
\[
\dist(\bx,T):=\min_{\bsigma\in T}\|\bx-\bsigma\|_2.
\]
(We use the convention that $\dist(\bx,T)=\infty$ when $T=\emptyset$.)

\begin{definition}[Isolated solution]
Fix an integer $k\in\{1,\dots,N\}$.
A solution $\bsigma\in\cS(\bG,\kappa)$ is \emph{$k$-isolated} if
\[
\cS(\bG,\kappa)\cap \Ball_{k,\mathrm{H}}(\bsigma)=\{\bsigma\}.
\]
We write $\mathcal{S}^\circ_k(\bG,\kappa)$ for the set of $k$-isolated solutions. Note that if a solution $\bsig$ is $k$-isolated then it is also $k'$-isolated for any $k'\leq k$. For a fraction $\iota\in(0,1)$, we also write $\mathcal{S}^\circ (\bG,\kappa,\iota):=\mathcal{S}^\circ_{\lfloor \iota N\rfloor}(\bG,\kappa)$.
\end{definition}

To characterize this notion of isolation more precisely, we also use the notion of \emph{margin}, which measures how close a solution is to the boundary of the constraints.
\begin{definition}[Margin]
For $\bsigma\in\bSigma_N$, define
\[
m(\bsigma)=m_{\kappa}(\bsigma;\bG) := \min_{1\le a\le M}\ \frac{\la \bg^a,\bsigma\ra}{\sqrt{N}}-\kappa\in\R.
\]
Then $\bsigma\in\cS(\bG,\kappa)$ if and only if $m(\bsigma)\ge 0$.
\end{definition}

For a resampling (``noise'') parameter $\eta_N\in(0,1)$, define a correlated copy of $\bG$ by
\begin{equation}
\label{eq:noise}
\widetilde 
\bG := \sqrt{1-\eta_N}\,\bG + \sqrt{\eta_N}\,\bG',
\end{equation}
where $\bG'$ is an independent copy of $\bG$.
Then $\widetilde \bG\stackrel{d}{=}\bG$.

Below, we measure stability and success in $\ell_2$. Since we do not require the algorithm output $\mathcal{A}_N(\bG)$ to lie exactly on the hypercube $\bSig_N$, $\ell_2$ is a natural metric. Moreover, when both points lie on $\bSig_N$, $\ell_2$ exactly captures the Hamming distance: for $\btau,\bsig\in\bSig_N$,
\[
\|\bsig-\btau\|_2^2 = 4\, d_{\mathrm H}(\bsig,\btau).
\]
The next definition provides the stability condition.
\begin{definition}[$\ell_2$-stability under resampling]
Fix a sequence $\eta_N\in(0,1)$.
A (possibly randomized) algorithm $\mathcal A_N: \R^{N\times M}\times \Omega\to[-1,1]^N$ is \emph{$(\rho_N,t_N)$-stable at noise level $\eta_N$} if, for $(\bG,\widetilde \bG)$ coupled as in \eqref{eq:noise} and an independent shared random seed $\omega\in\Omega$,
\[
\bbP\big[\|\mathcal A_N(\bG,\omega)-\mathcal A_N(\widetilde \bG,\omega)\|_2 > \rho_N\big] \le t_N.
\]
\end{definition}
Corollary~\ref{cor:ld-stable} shows that degree-$D$ polynomial outputs with bounded second moment satisfy $\bbE\|\mathcal A_N(\bG)-\mathcal A_N(\widetilde \bG)\|_2^2\lesssim D\eta_N N$, so they are stable on the scales used below.

\begin{definition}[Locating solutions in $\ell_2$]
\label{def:locating}
Fix $\iota\in(0,1)$.
We say $\mathcal A_N$ \emph{$\iota N$-locates a solution} if
\[
\dist\big(\mathcal A_N(\bG),\cS(\bG,\kappa)\big) \le \frac{\sqrt{\iota N}}{3}.
\]
We say $\mathcal A_N$ \emph{locates an $\iota N$-isolated solution} if
\[
\dist\big(\mathcal A_N(\bG),\cS^\circ(\bG,\kappa,\iota)\big) \le \frac{\sqrt{\iota N}}{3}.
\]
More generally, for an integer $k$ we say $\mathcal A_N$ \emph{$k$-locates a $k$-isolated solution} if
\[
\dist\big(\mathcal A_N(\bG),\cS^\circ_k(\bG,\kappa)\big) \le \frac{\sqrt{k}}{3}.
\]
\end{definition}

\subsubsection{Stable algorithms cannot find isolated solutions}

Our first result shows that stability alone precludes success probability approaching $1$.
The argument yields an explicit constant.

\begin{theorem}
\label{thm:stable-v1}
Fix any $\kappa\in\R$, $\alpha>0$, and $\iota\in(0,1)$.
Let $\eta_N = o_N(1)$ satisfy
\begin{equation}
\label{eq:eta-condition}
\eta_N = \omega \left( \frac{\log N}{N}\right)\,.
\end{equation}
Let $\mathcal A_N$ be $(\rho_N,t_N)$-stable at noise level $\eta_N$ with
$
\rho_N=o(\sqrt{N}),\,t_N=o_N(1).
$
Then
\[
\bbP\big[\mathcal A_N \text{ locates an $\iota N$-isolated solution of }(\bG,\kappa)\big]
\le \frac{3\sqrt{17}-9}{4}+o_N(1).
\]
\end{theorem}

\begin{remark}[Where the constant $\frac{3\sqrt{17}-9}{4}$ comes from]
The constant $\frac{3\sqrt{17}-9}{4}$ is the positive root of 
\[
\frac{2s^2}{9} = 1-s.
\]
It is the threshold at which the Pitt lower bound (quadratic in a partition mass) overtakes the trivial upper bound $1-s$ coming from the event ``there is exactly one nearby solution''. We refer to the technical overview in Subsection~\ref{subsec:proof-overview} for further discussion.
\end{remark}
\begin{remark}[On the stability condition]
This stability condition requires that if the disorder is perturbed at scale $\sqrt{\eta_N}$, then the corresponding outputs change by only $o(\sqrt{N})$ in $\ell_2$ distance. This is a natural and fairly mild requirement, since the $\ell_2$ distance between any two points in $[-1,1]^N$ is at most of order $\sqrt{N}$. As discussed in Section~\ref{subsubsec:high-success-algorithms}, the algorithm of \cite{kim1998covering} satisfies this condition. Moreover, as shown in Section~\ref{subsubsection:low-degree-hardness}, this stability condition can be transferred to hardness for all degree-$o(N/\log N)$ polynomials, which is nearly tight up to a logarithmic factor (this suggests $\exp(\wt\Theta(N))$ running time is required to reliably find a strongly isolated solution, while random guessing takes $\exp(\Theta(N))$ time). It would also be interesting to determine whether such a stability condition holds for the algorithms in \cite{abbe2022binary,li2025discrepancy,bansal2020line}. We expect that all of them are stable in this sense.
\end{remark}

\subsubsection{High-success algorithms avoid isolated solutions}
\label{subsubsec:high-success-algorithms}
Our second theorem shows that if a stable algorithm solves the perceptron with high probability, then the probability that it locates an isolated solution must be negligible.

\begin{theorem}
\label{thm:stable-v2}
Fix $\kappa\in\R$, $\alpha>0$, and $\iota\in(0,1)$.
Let $\eta_N = o_N(1)$ satisfy \eqref{eq:eta-condition}, and let $\mathcal A_N$ be $(\rho_N,t_N)$-stable at noise level $\eta_N$, with $\rho_N=o(\sqrt{N})$ and $t_N=o_N(1)$.
Then for every $\delta>0$, if
\[
\bbP\big[\mathcal A_N \text{ $\iota N$-locates a solution of } (\bG,\kappa)\big] \ge 1-\delta,
\]
then
\[
\bbP\big[\mathcal A_N \text{ locates an $\iota N$-isolated solution of } (\bG,\kappa)\big] \leq o_\delta(1),
\]
where $o_\delta(1)$ denotes a quantity depending on $\delta$ but not on $N$, and satisfying $o_\delta(1)\to 0$ as $\delta\to 0$.
\end{theorem}

\begin{remark}
    Consider the multi-scale majority algorithm in \cite{kim1998covering}. It is proved in \cite[Theorem~3.8]{gamarnik2022algorithms} that this algorithm is stable in the sense of Theorem~\ref{thm:stable-v2}. Since the algorithm succeeds in finding a solution with high probability, Theorem~\ref{thm:stable-v2} implies that, with high probability, the solution it finds is not in $\cS^\circ(\bG,\kappa,\iota)$.
\end{remark}

\subsubsection{Low-degree consequence}
\label{subsubsection:low-degree-hardness}

\begin{definition}[Degree-$D$ polynomial output with bounded second moment]
\label{def:poly}
Fix $D\in\bbN$.
A \emph{degree-$D$ polynomial output} is a (possibly randomized) map $\mathcal A_N^\circ:\R^{N\times M}\times\Omega\to\R^N$ such that:
\begin{itemize}[leftmargin=2em]
\item For fixed $\omega\in\Omega$, each coordinate of $\mathcal A_N^\circ(\bG,\omega)$ is a polynomial in the entries of $\bG$ of degree at most $D$;
\item there is a constant $C$ (independent of $N$) with $\bbE[\|\mathcal A_N^\circ(\bG)\|_2^2]\le C N$.
\end{itemize}
The second ``random seed'' argument $\omega\in\Omega$ will always be taken independent of $\bG$; it is not crucial but simply makes the framework slightly more general.
We define the associated bounded-output algorithm $\mathcal A_N(\bG):=\proj(\mathcal A_N^\circ(\bG))\in[-1,1]^N$, where $\proj$ is the coordinate-wise projection to $[-1,1]$.
\end{definition}

The following standard stability estimate, proved in Section~\ref{sec:prelim}, transfers our results on stable algorithms to low-degree polynomials.
\begin{corollary}[Degree-$D$ polynomial outputs are stable]
\label{cor:ld-stable}
Let $\mathcal A_N=\proj\circ \mathcal A_N^\circ$ where $\mathcal A_N^\circ$ is degree-$D$ with $\bbE\|\mathcal A_N^\circ(\bG)\|_2^2\le C N$.
For $(\bG,\widetilde \bG)$ coupled as in \eqref{eq:noise} at noise level $\eta$,
\[
\bbE\big[\|\mathcal A_N(\bG)-\mathcal A_N(\widetilde \bG)\|_2^2\big] \le 2 C D\,\eta_N\,N.
\]
Consequently, if $D\eta_N =o(1)$ then $\mathcal A_N$ is $(\rho_N,t_N)$-stable for some $\rho_N=o(\sqrt N)$ and $t_N = o(1)$.
More generally, if $D\eta_N N=o(k)$ then $\rho_N=o(\sqrt{k})$.
\end{corollary}

Corollary~\ref{cor:ld-stable} gives $(\rho_N,t_N)$-stability with $\rho_N=o(\sqrt N)$ whenever $D_N\eta_N=o(1)$. Combining this with Theorem~\ref{thm:stable-v1}, one obtains the following consequence. If $D_N=o\!\left(\frac{N}{\log N}\right)$, $\mathcal A_N^\circ$ is a degree-$D_N$ polynomial output in the sense of Definition~\ref{def:poly}, and $\mathcal A_N=\proj\circ \mathcal A_N^\circ$, then
\[
\bbP\big[\mathcal A_N \text{ locates an $\iota N$-isolated solution of }(\bG,\kappa)\big]
\le \frac{3\sqrt{17}-9}{4}+o_N(1).
\]

\begin{remark}
    Under the low-degree hypothesis \cite{hopkins2018statistical}, the failure of degree-$\frac{N}{\log N}$ polynomials provides evidence that an exponential running time of order $\exp\big(\frac{N}{\mathrm{polylog}N}\big)=\exp(N^{1-o_N(1)})$ is necessary for an algorithm to locate an $\iota N$-isolated solution of $(\bG,\kappa)$ with high probability. Note that random guessing takes $\exp(\Theta(N))$ time to find a solution. Therefore, the power of $N$ in the exponent is tight.
\end{remark}

\begin{remark}
    Although we consider a fixed $\kappa$ for convenience, if each constraint $\bg^a$ is assigned a uniformly bounded threshold $\kappa_a\in [-C_0,C_0]$ independently of the disorder, none of our arguments change.
    One can also reverse the sign of some constraints, i.e. require $\la \bg^a,\bsig\ra\leq\kappa_a\sqrt{N}$ for $a\in A\subseteq [M]$; this is equivalent to negating $\bg^a$ which does not change the problem since $\bg^a\stackrel{d}{=}-\bg^a$.
\end{remark}

We also record variants on subextensive Hamming scales and, in particular, for perceptrons defined by a finite union of intervals. When $k_N=o(N)$, the interpretation of ``$k_N$-locating a $k_N$-isolated solution'' is less canonical than in the macroscopic setting, since one could imagine following an $o(N)$-accurate output by a simple clean-up phase that corrects the remaining coordinates. In the finite-union-of-interval setting, the argument below applies only in the regime $k_N\le \sqrt{N}/(\log N)^2$.

With this caveat, we have the following scale-dependent statement.
\begin{theorem}
\label{thm:k-isolation}
Fix $\kappa\in\R$, $\alpha>0$, and let $k=k_N$ be an integer sequence with $k_N\to\infty$ and $k_N\le N/10$.
Let $\eta_N \leq \frac{(\log N)^{1.1}}{N}$ satisfy \eqref{eq:eta-condition}.
Let $\mathcal A_N:\R^{N\times M}\times \Omega\to[-1,1]^N$ be $(\rho_N,t_N)$-stable at noise level $\eta_N$ with
$
\rho_N=o(\sqrt{k_N}),t_N=o_N(1).
$
Then
\[
\bbP\big[\mathcal A_N \text{ $k_N$-locates a $k_N$-isolated solution of }(\bG,\kappa)\big]
\le \frac{3\sqrt{17}-9}{4}+o_N(1).
\]
\end{theorem}
Corollary~\ref{cor:ld-stable} yields the same conclusion for degree-$D_N$ polynomial outputs with $D_N=o\!\left(\frac{k_N}{\log N}\right)$ and bounded second moment. In particular, degree-$\mathrm{polylog}(N)$ polynomial outputs cannot locate $\mathrm{polylog}(N)$-isolated solutions with success probability exceeding $\frac{3\sqrt{17}-9}{4}+o_N(1)$.

Next, for any $\mathcal U\subset \mathbb R$, we define a generalized set of solutions to be 
\begin{equation}
    \label{eq:sol-space-SBP}
    \mathcal{S}(\bG,\mathcal{U}):=\bigcap_{a=1}^M \Big\{\bsig\in\bSig_N:\, \langle \bg^a,\bsig\rangle/\sqrt{N}\in \mathcal{U}\Big\}.
\end{equation}
Similarly, we define the set of isolated solutions $\mathcal{S}^\circ_k(\bG,\mathcal{U})$ and $\mathcal{S}^\circ(\bG,\mathcal{U},\iota)$. In this subsection, we assume that $\mathcal{U}=\sqcup_{i=1}^L\mathcal{I}_i$, where $L$ is a finite number and each $\mathcal{I}_i$ is an interval.
For such perceptrons one obtains the following local analogue.
\begin{theorem}
    \label{thm:stable-SBP}
    Fix $\alpha>0$ and $\mathcal U$ as above, and let $k_N$ be an integer with $k_N\to\infty$ and $k_N\leq \sqrt{N}/(\log N)^2$. Let  $\eta_N \leq \frac{(\log N)^{1.1}}N$ satisfy \eqref{eq:eta-condition}. Let $\mathcal{A}_N:\mathbb R^{N\times M}\times\Omega\to[-1,1]^N$ be $(\rho_N,t_N)$-stable at noise level $\eta_N$ with 
    $\rho_N=o(\sqrt{k_N})\,, t_N=o_N(1)$.
    Then
    \begin{equation}
    \label{eq:conc-thm-1.10}
        \mathbb P[\mathcal{A}_N \text{ $k_N$-locates a $k_N$-isolated solution of }(\bG,\mathcal{U})]\leq \frac{3\sqrt{17}-9}{4}+o_N(1)\,.
    \end{equation}
\end{theorem}

The same argument also yields an analogue of Theorem~\ref{thm:stable-v2} for perceptrons defined by a finite union of intervals, which we omit here. Additionally, Corollary~\ref{cor:ld-stable} implies that if $D_N=o\!\left(\frac{k_N}{\log N}\right)$ and $k_N\le \sqrt{N}/(\log N)^2$, then degree-$D_N$ polynomial outputs with bounded second moment satisfy \eqref{eq:conc-thm-1.10}.

\subsection{Related work}
Being concerned with algorithmic hardness, our result is related to two major approaches to such questions: the \emph{overlap gap property (OGP)} \cite{Garmarnik21} and the \emph{low-degree polynomial framework} \cite{wein2025computational}.

\subsubsection{Overlap gap property (OGP)}
\label{ss:ogp}

There is a long history of work on the solution geometry of disordered problems \cite{mezard1984nature,mezard2005clustering,achlioptas2006solution,krzakala2007gibbs,achlioptas2008algorithmic}, and it has long been expected that solution geometry has implications for algorithmic tractability. The overlap gap property, introduced in \cite{gamarnik2014limits}, provided an important framework to rigorously link solution geometry and computational hardness. 
Since then, OGP has become a central tool for studying computational hardness in random search and optimization problems; see, for example, \cite{RV17,GL18,wein2022optimal,gamarnik2022algorithms,bresler2022algorithmic,gamarnik2023geometric,huang2025tight}. In its simplest form, an OGP roughly asserts that the overlap between two solutions cannot lie in a certain intermediate range, so that any two solutions must be either far apart or close together. This geometric obstruction has proved to be a powerful barrier for stable algorithms. Several refinements of the original notion have since been developed, including multi-OGP \cite{RV17,gamarnik2017performance,gamarnik2023algorithmic}, ladder OGP \cite{wein2022optimal,bresler2022algorithmic}, and branching OGP \cite{huang2025tight,jones2022random,huang2023algorithmic,du2025algorithmic,bhamidi2026finding,Huang2026perceptronalg}. More recently, OGP-based ideas have also been used to rule out \emph{online} algorithms; see, for example, \cite{gamarnik2023geometric,du2025algorithmic,gamarnik2025optimal,dhawan2025sharp,bhamidi2026finding}. Relatedly, recent work has also shown that the connection between OGP and algorithmic hardness can be subtle: \cite{li2025some} exhibits algorithmically easy average-case optimization problems that nevertheless satisfy OGP, while \cite{huang2025strong} provides a refined rule of thumb for converting OGPs into algorithmic hardness, and \cite{ma2025polynomial} gives a related counterexample for sampling.

Because OGP rules out stable algorithms through geometric properties of the solution space, it is natural to ask whether the same framework can also rule out algorithms that locate isolated solutions in perceptron models.
However, it is not clear how to make such an approach rigorous. Our proof instead takes a different route that does not rely on OGP; we are not aware of prior work showing hardness for a search problem against stable algorithms without an OGP-based argument. We refer to Subsection~\ref{subsec:proof-overview} for further discussion.

\subsubsection{Low-degree polynomials}
A complementary approach is the low-degree polynomial framework. Originally developed in the study of high-dimensional inference problems, it has since become a general framework for analyzing average-case computational problems. In many such settings, low-degree polynomials capture the best known guarantees achieved by algorithmic families such as spectral methods, approximate message passing, and subgraph counts; see, for example, \cite{hopkins2018statistical,schramm2022computational,mao2024testing,ding2025low,chen2026computational,bok2026detection}. Moreover, the framework also provides strong evidence for algorithmic hardness. Indeed, a conjecture of Hopkins \cite{hopkins2018statistical} predicts that the failure of degree-$D$ polynomial algorithms implies the failure of ``robust'' algorithms running in time $n^{\wt O(D)}$. This is often referred to as the \emph{low-degree conjecture}. 

In recent years, the low-degree framework has also been brought to bear on \emph{random optimization problems} \cite{GJW20,gamarnik2024hardness}, including the largest independent set problem \cite{wein2022optimal}, random $k$-SAT \cite{bresler2022algorithmic}, and local optima in spin glasses \cite{huang2025strong}. 
Here (unlike the statistical estimation setting) the low-degree property is relevant mainly due to its stability implications.

\subsection{Technical overview}
\label{subsec:proof-overview}
Here we outline the proofs of our results and discuss the key challenges in applying the OGP framework.

\paragraph{Stable algorithms cannot find isolated solutions.}
A standard route to hardness results for stable algorithms is the overlap-gap framework. In our setting, however, that approach appears to face substantial technical obstacles; we discuss these difficulties in a later paragraph. Instead, we introduce a completely different argument, which to our knowledge has not appeared in previous work.

The starting point is to exploit the stability assumption directly. We consider two correlated Gaussian disorder matrices $\bG$ and $\wt{\bG}$, with correlation $\sqrt{1-\eta_N}$, where $\eta_N$ is very small, typically of order $\omega(\log N/N)$. We argue by contradiction. Suppose that there exists a stable algorithm $\mathcal A_N$ which locates an $\iota N$-isolated solution with large probability, say $0.99$. 
To simplify the discussion, we describe the simpler case in which the algorithm outputs an isolated solution exactly; the same idea also handles our more permissive success notion, where the algorithm is only required to output a point close to an isolated solution.
Consider the $\ell_2$-ball $\mathsf B_{\sqrt{\iota N}/3,2}\bigl(\mathcal A_N(\bG)\bigr)$ centered at the output of the algorithm on input $\bG$ (see the purple ball in Figure~\ref{fig:proof idea}). Conditioning on $\bG$, let $Z$ denote the number of solutions to $\wt{\bG}$ inside this ball.

\begin{figure}[ht]
\centering
\begin{tikzpicture}[scale=1.08]

\definecolor{ballpurple}{RGB}{164,132,255}
\definecolor{ballrose}{RGB}{244,201,210}
\definecolor{ballgold}{RGB}{236,206,120}
\definecolor{ballteal}{RGB}{142,206,200}
\definecolor{textgray}{RGB}{55,65,81}

\begin{scope}[shift={(0,0)}]

\coordinate (AG)  at (0,0);
\coordinate (AtG) at (0.78,0.36);

\fill[ballrose, opacity=0.28] (AtG) circle [radius=3.24];

\fill[ballpurple, opacity=0.72] (AG) circle [radius=1.5];

\fill[textgray] (AG)  circle [radius=1.2pt];
\fill[textgray] (AtG) circle [radius=1.2pt];

\node[anchor=north east, xshift=6mm, yshift=-1mm] at (AG) {$\mathcal A_N(\bG)$};
\node[anchor=north, xshift=1mm, yshift=0.5mm] at (AtG) {$\mathcal A_N(\widetilde{\bG})$};

\end{scope}

\begin{scope}[shift={(8.2,0)}]

\def\r{1.5}

\fill[ballgold, opacity=0.88] (0,0) circle (\r);

\begin{scope}
  \clip (0,0) circle (\r);
  \fill[ballteal, opacity=0.88]
    (0.22,1.55)
      .. controls (-0.20,1.02) and (-0.03,0.40) ..
    (-0.08,0.00)
      .. controls (-0.12,-0.38) and (0.10,-0.95) ..
    (0.32,-1.55)
      -- (2,-1.55) -- (2,1.55) -- cycle;
\end{scope}

\draw[textgray, line width=0.55pt, opacity=0.45]
    (0.22,1.55)
      .. controls (-0.20,1.02) and (-0.03,0.40) ..
    (-0.08,0.00)
      .. controls (-0.12,-0.38) and (0.10,-0.95) ..
    (0.32,-1.55);

\draw[textgray, opacity=0.18, line width=0.3pt] (0,0) circle (\r);

\coordinate (sig) at (-0.42,0.08);
\coordinate (tau) at (0.42,0.00);

\fill[textgray] (sig) circle [radius=1.2pt];
\fill[textgray] (tau) circle [radius=1.2pt];

\node[above left=1pt] at (sig) {$\bsigma$};
\node[above right=1pt] at (tau) {$\btau$};

\end{scope}

\end{tikzpicture}
\caption{Left: the ball $\mathsf B_{\sqrt{\iota N}/3,2}\bigl(\mathcal A_N(\bG)\bigr)$ around the output $\mathcal A_N(\bG)$, together with a larger neighborhood around $\mathcal A_N(\widetilde{\bG})$. Right: schematically, the ball $\mathsf B_{\sqrt{\iota N}/3,2}\bigl(\mathcal A_N(\bG)\bigr)$ is divided into a left region containing $\bsigma$ and a right region containing $\btau$.}
\label{fig:proof idea}
\end{figure}
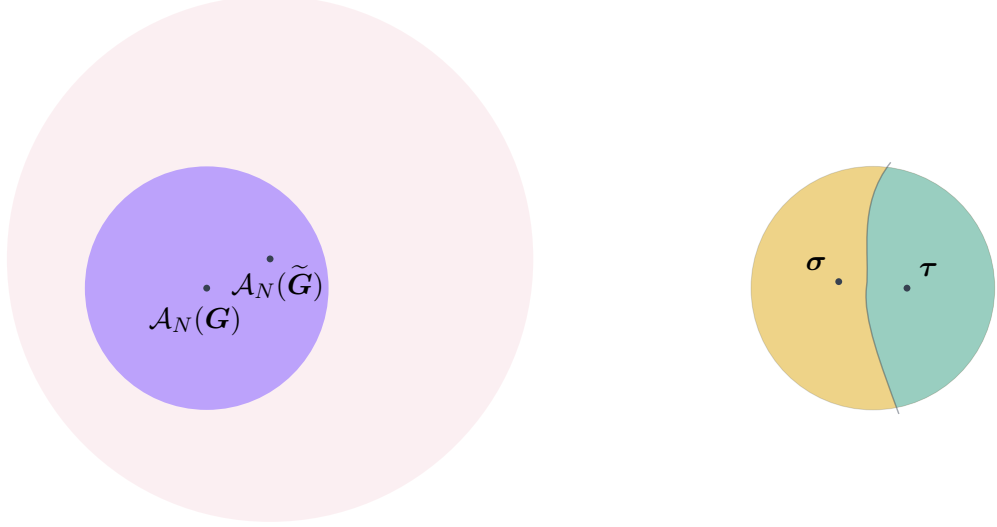

The key point is that, under our assumption, with large probability one has $Z=1$. Indeed, stability implies that $\mathcal A_N(\wt{\bG})$ remains close to $\mathcal A_N(\bG)$, so whenever $\mathcal A_N(\wt{\bG})$ successfully locates a solution, that solution must lie inside the above ball, giving $Z\ge 1$. On the other hand, if $\mathcal A_N(\wt{\bG})$ locates an $\iota N$-isolated solution, then there cannot be any other solution within Hamming distance $\iota N$ of it. Since the above $\ell_2$-ball is contained in such a Hamming neighborhood, it follows that the ball contains at most one solution, hence $Z\le 1$.

We then show that this scenario is essentially impossible. To do so, we partition the small ball into two disjoint parts, 
$\mathsf C_1$ and $\mathsf C_2$, and let $Z_1,Z_2$ denote the numbers of solutions to $\wt{\bG}$ in these two regions. The conclusion that $Z_1+Z_2=Z=1$ suggests a negative correlation between $Z_1$ and $Z_2$: if one side contains a solution, then the other typically cannot. We prove instead that the opposite is true: the relevant solution events are in fact non-negatively correlated, leading to a contradiction.

The reason is that for any two candidate vectors $\bsigma\in \mathsf C_1$ and $\btau\in \mathsf C_2$, both vectors lie in a very small neighborhood, and hence are close to each other. The events that $\bsigma$ and $\btau$ are solutions to $\wt{\bG}$ are determined by the inequalities
\[
\langle \wt{\bg}^a,\bsigma\rangle \ge \kappa\sqrt N,
\qquad
\langle \wt{\bg}^a,\btau\rangle \ge \kappa\sqrt N,
\qquad a\in[M].
\]
Because $\bsigma$ and $\btau$ are close, these Gaussian random variables $\langle \wt{\bg}^a,\bsigma\rangle$ and $\langle \wt{\bg}^a,\btau\rangle$ are positively correlated, whereas for $a\neq b$, the random variables $\langle \wt{\bg}^a,\bsigma\rangle$ and $\langle \wt{\bg}^b,\btau\rangle$ are conditionally independent given $\bG$.
Pitt's correlation inequality\footnote{In fact the better-known FKG inequality can also be used instead of Pitt's inequality, after a suitable orthogonal transformation on $\bbR^N$ (which makes no difference for isotropic Gaussian disorder).
The point is to make all vectors in a small Hamming neighborhood of say $\bsig=(1,1,\dots,1)$ have all their coordinates positive. 
This can be ensured by reflecting the Boolean cube through the line $\bbR\bsig$, after which all $\btau$ within Hamming distance $N/4$ of $(1,1,\dots,1)$ have positive coordinates.
The remaining argument is then identical to ours.
} on Gaussian space then implies that the corresponding feasibility events are also non-negatively correlated. Intuitively, if two candidate configurations are very similar, then the event that one satisfies all constraints should only increase the chance that the other does as well.

The step of splitting the small ball into two parts involves some technicalities. To make the negative-correlation argument work, we need the partition to be approximately balanced. For this, we show that after a small perturbation, each binary vector in a small Hamming neighborhood of the isolated solution is a solution to $\wt\bG$ with probability at most $1/2+o_N(1)$. A discrete intermediate value argument (see Lemma~\ref{lem:partition-I}) then suggests that this neighborhood can be partitioned into two subsets such that each subset contains a solution for $\wt\bG$ with probability bounded away from zero.
More precisely, to explain the origin of the constant $\frac{3\sqrt{17}-9}{4}$, let $S$ be the probability that the small ball contains an isolated solution. By Lemma~\ref{lem:partition-I}, we may partition the ball into $\mathsf C_1$ and $\mathsf C_2$ so that the product of the probabilities that $\mathsf C_1$ and $\mathsf C_2$ each contain a solution is at least $2S^2/9$. By Pitt's inequality, the probability that both $\mathsf C_1$ and $\mathsf C_2$ contain a solution, and hence that the ball contains at least two solutions, is at least $2S^2/9$. Therefore, $\frac{2S^2}{9}\le 1-S$, which yields $S\le \frac{3\sqrt{17}-9}{4}$.

In summary, stability together with successful recovery of isolated solutions would force a small neighborhood to contain exactly one solution with high probability. After partitioning this neighborhood into two pieces, this would require the solution counts in the two pieces to be negatively correlated. Pitt's inequality rules this out by giving a non-negative correlation instead, yielding the desired contradiction. This obstruction is fundamentally different from previous hardness arguments based on overlap gaps.

\paragraph{Low-degree consequence.}
The low-degree lower bounds are obtained by a reduction to the stable-algorithm setting. Lemma~3.4 in \cite{gamarnik2024hardness} shows that if $\mathcal A_N^\circ$ is a degree-$D$ polynomial output with bounded second moment, then under a small Gaussian resampling of the disorder one has
\[
\mathbb E\bigl[\|\mathcal A_N^\circ(\bG)-\mathcal A_N^\circ(\widetilde \bG)\|_2^2\bigr]\lesssim D\eta_N N.
\]
Since coordinate-wise clipping to $[-1,1]$ is nonexpansive, the associated bounded output $\mathcal A_N=\Pi\circ \mathcal{A}_N^\circ$ inherits the same stability estimate. Thus, whenever $D\eta_N\to 0$, degree-$D$ polynomial algorithms fall within the scope of our stable-algorithm obstruction. In this way, the impossibility result for stable algorithms transfers directly to low-degree polynomials, yielding the consequence stated in Subsection~\ref{subsubsection:low-degree-hardness}.

\paragraph{Symmetric binary and generalized perceptrons.}
A similar argument also applies to the symmetric binary perceptron. The main new issue in the symmetric case is that each constraint is two-sided, so the global feasibility event is no longer monotone in the Gaussian constraint values, unlike in the asymmetric setting. As a result, Pitt's inequality cannot be invoked directly.
However this non-monotonicity disappears at a sufficiently local scale. Indeed, within any Hamming ball of radius $o(\sqrt{N})$, and for any fixed pattern $\bg^a$, at most one of the two boundary inequalities can be active throughout the ball: traversing both boundaries would require the corresponding inner product to change by $\Omega(1)$, which is impossible on such a small neighborhood. Therefore, after conditioning on the relevant side of each constraint for each $a$, the local feasibility events become coordinate-wise monotone, and Pitt's correlation inequality again applies. The same principle extends to perceptron models defined by a finite union of intervals.

\paragraph{Challenges and limitations of the OGP framework.}
The standard route to hardness for stable algorithms is through the overlap-gap property (OGP) \cite{Garmarnik21}. Consider the $2$-OGP as a simple illustration. For a fixed instance $\bG$, suppose there exist two strongly isolated $\bsig_1,\bsig_2\in \mathcal{S}^\circ(\bG,\kappa,\iota)$. Since both are isolated solutions of the same instance, their overlap is automatically constrained by isolation, namely
\[
\langle \bsig_1,\bsig_2\rangle \le (1-2\iota)N.
\]
However, to use the $2$-OGP to rule out stable algorithms, one must instead consider a pair of correlated instances $(\bG_1,\bG_2)$, where $\bsig_j$ is an isolated solution for $\bG_j$. In this setting, the overlap $\langle \bsig_1,\bsig_2\rangle$ is no longer controlled for free by the isolation condition, and must itself be analyzed.

This leads to a substantial technical difficulty: one is forced to study the local geometry around $\bsig_1$ under conditioning on which nearby spin configurations are feasible for the corresponding instance. Such conditioning introduces a complicated dependence structure that is difficult to handle. Moreover, even if one aims only for a stability-based contradiction, it is not clear what form of uniform control should remain available, since arguments involving correlated pairs $(\bsig_1,\bsig_2)$ typically require union bounds over the solution set, which can be prohibitively costly.

\subsection{Organization}
The rest of the paper is organized as follows. In Section~\ref{sec:prelim}, we collect the preliminary tools used throughout the paper, including correlation inequalities and the $\ell_2$-stability of low-degree polynomial outputs. In Section~\ref{sec:property}, we develop structural properties of isolated solutions, showing in particular that isolation forces small margin and that isolated solutions are fragile under small perturbations of the disorder. In Section~\ref{sec:proof}, we prove the stable-algorithm obstruction, the consequence that highly successful stable algorithms avoid isolated solutions, and the corresponding low-degree corollary; we then discuss the $k$-scale and finite-union-of-interval variants. Finally, in Section~\ref{sec:discuss}, we conclude with a discussion and several open problems.

\subsection{Notations}
We use standard big-$O$ notation throughout, with $N$ as the underlying asymptotic parameter. Thus, for sequences $a_N$ and $b_N$, we write $a_N = O(b_N)$ and $a_N \lesssim b_N$ to mean that $a_N \le C b_N$ for some absolute constant $C>0$; similarly, $a_N = \Omega(b_N)$ and $a_N \gtrsim b_N$ mean that $b_N = O(a_N)$, and $a_N = \Theta(b_N)$ means that both $a_N = O(b_N)$ and $a_N = \Omega(b_N)$. We also write $a_N = o(b_N)$ and $a_N \ll b_N$ if $a_N/b_N \to 0$ as $N\to\infty$, and $a_N = \omega(b_N)$ and $a_N \gg b_N$ if $b_N = o(a_N)$. Finally, $\widetilde O$, $\widetilde \Omega$, $\widetilde \Theta$, $\widetilde o$, and $\widetilde \omega$ hide factors that are polylogarithmic in $N$. When the asymptotic parameter is clear from context, we occasionally suppress the subscript $N$ in this notation.
We write $\sqcup$ for disjoint union.

\subsection{Preliminaries: correlation inequalities and low-degree stability}\label{sec:prelim}

We will use the following classical correlation inequality of Pitt (which is proved via interpolation).

\begin{proposition}[\cite{pitt1982positively}]
\label{prop:pitt}
Let $F,G:\R^D\to\R$ be coordinate-wise increasing functions, and let $\bX\in\R^D$ be a (not necessarily centered) Gaussian vector whose covariances are all nonnegative: $\mathrm{Cov}(\bX_i,\bX_j)\ge 0$ for all $i,j$.
Then
\[
\bbE[F(\bX)G(\bX)] \ge \bbE[F(\bX)]\,\bbE[G(\bX)].
\]
The same conclusion holds if $F$ and $G$ are both coordinate-wise \emph{decreasing}.
\end{proposition}

We next deduce the low-degree stability estimate stated earlier as Corollary~\ref{cor:ld-stable}; it is immediate from the two statements below and Markov's inequality.

\begin{proposition}[Noise sensitivity of low-degree polynomials (see Lemma~3.4, \cite{gamarnik2024hardness})] 

\label{prop:gjw}
Let $\mathcal A_N^\circ$ be a degree-$D$ polynomial output with $\bbE[\|\mathcal A_N^\circ(\bG)\|_2^2]\le C N$.
For $(\bG,\widetilde \bG)$ coupled as in \eqref{eq:noise} at noise level $\eta_N$, 
\[
\bbE\big[\|\mathcal A_N^\circ(\bG)-\mathcal A_N^\circ(\widetilde \bG)\|_2^2\big] \le 2 C D\,\eta_N\, N.
\]
\end{proposition}

\begin{lemma}
\label{lem:clip}
The coordinate-wise projection $\proj:\R^N\to[-1,1]^N$ satisfies
\[
\|\proj(\bx)-\proj(\by)\|_2 \le \|\bx-\by\|_2\qquad \text{for all }\bx,\by\in\R^N.
\]
In particular, for $\mathcal A_N=\proj\circ \mathcal A_N^\circ$,
\[
\bbE\|\mathcal A_N(\bG)-\mathcal A_N(\widetilde \bG)\|_2^2
\le \bbE\|\mathcal A_N^\circ(\bG)-\mathcal A_N^\circ(\widetilde \bG)\|_2^2.
\]
\end{lemma}

\begin{proof}
Since $u\mapsto \min\{1,\max\{-1,u\}\}$ is $1$-Lipschitz on $\R$, the map $\Pi$ is $1$-Lipschitz on $\R^N$ in $\ell_2$.
\end{proof}

\section{Properties of isolated solutions}\label{sec:property}
This section provides the key structural properties of isolated solutions that drive our impossibility results.

\subsection{Isolation forces small margin}

\begin{proposition}
\label{prop:isolation-implies-low-margin}
Fix $\kappa\in\R$ and $\alpha>0$.
With probability $1-o_N(1)$ over $\bG$, every $\bsigma\in\cS^\circ_k(\bG,\kappa)$ (for every $k\ge 1$) satisfies
\[
m(\bsigma) \le \frac{2\,\|\bG\|_{\infty}}{\sqrt{N}}, \qquad \text{where }\ \|\bG\|_{\infty}:=\max_{1\le i\le N,\ 1\le a\le M} |\bg_i^a|.
\]
In particular, since $M=\Theta(N)$, one has $\|\bG\|_{\infty}=O\big(\sqrt{\log N}\big)$ with probability $1-o_N(1)$, hence
\[
m(\bsigma)=O\Big(\sqrt{\frac{\log N}{N}}\Big)
\quad \text{uniformly over all }\bsigma\in\bigcup_{k\ge 1}\cS^\circ_k(\bG,\kappa).
\]
\end{proposition}

\begin{proof}
Fix $\bG$ and suppose $\sigma\in\cS^\circ_k(\bG,\kappa)$ with $k\ge 1$.
For each coordinate $i$, let $\bsigma^{(i)}$ denote the configuration obtained from $\bsigma$ by flipping its $i$-th bit.
Then $\dH(\bsigma,\bsigma^{(i)})=1\le k$, so by $k$-isolation we have $\bsigma^{(i)}\notin\cS(\bG,\kappa)$.
Hence for each $i$ there exists an index $a(i)\in\{1,\dots,M\}$ with
\[
\frac{\la \bg^{a(i)},\bsigma^{(i)}\ra}{\sqrt{N}} < \kappa.
\]
But
\[
\la \bg^{a(i)},\bsigma^{(i)}\ra = \la \bg^{a(i)},\bsigma\ra - 2 \bg^{a(i)}_i\bsigma_i,
\]
so
\[
\frac{\la \bg^{a(i)},\bsigma\ra}{\sqrt{N}} - \kappa < \frac{2|\bg^{a(i)}_i|}{\sqrt{N}} \le \frac{2\|\bG\|_{\infty}}{\sqrt{N}}.
\]
Taking the minimum over $a$ on the left-hand side yields $m(\bsigma)\le 2\|\bG\|_{\infty}/\sqrt{N}$.
Finally, $\|\bG\|_{\infty}=O(\sqrt{\log N})$ with probability $1-o_N(1)$ by a union bound and standard Gaussian tails.
\end{proof}

\subsection{Fragility of isolated solutions under perturbation}
The next corollary shows that isolated solutions are unstable under a small Gaussian perturbation of the disorder.
\begin{corollary}
\label{cor:noise-destroys}
Fix $\kappa\in\R$, $\alpha>0$, and let $k\ge 1$.
Let $\eta_N = o(1)$ satisfy \eqref{eq:eta-condition}, and let $(\bG,\widetilde \bG)$ be coupled as in \eqref{eq:noise}.
Then with probability $1-o_N(1)$ over $\bG$, the following holds.
For every $\bsigma\in\cS^\circ_k(\bG,\kappa)$ and every $\btau\in\Ball_{k,\mathrm{H}}(\bsigma)$,
\begin{equation}
\label{eq:uniform-half}
\bbP\big[\btau\in\cS(\widetilde \bG,\kappa)\,\big|\,\bG\big] \le \tfrac12+o_N(1).
\end{equation}
\end{corollary}

\begin{proof}
Fix $\bG$ and suppose $\bsigma\in\cS^\circ_k(\bG,\kappa)$.
By definition, $\cS(\bG,\kappa)\cap \Ball_{k,\mathrm{H}}(\bsigma)=\{\bsigma\}$.

Let $a_\star\in\arg\min_a \big(\la \bg^a,\bsigma\ra/\sqrt{N}\big)$ be an index achieving the margin, so
\[
\frac{\la \bg^{a_\star},\bsigma\ra}{\sqrt{N}} = \kappa + m(\bsigma).
\]
Conditional on $\bG$, we have
\[
\frac{\la \widetilde \bg^{a_\star},\bsigma\ra}{\sqrt{N}} = \sqrt{1-\eta_N}\,(\kappa+m(\bsigma)) + \sqrt{\eta_N}\frac{\la \bg'^{a_\star},\bsig\ra}{\sqrt{N}}.
\]
Therefore by conditional independence,
    \begin{equation}\label{eq:singleton-after-perturbation}
\bbP\big[\bsigma\in\cS(\widetilde \bG,\kappa)\,\big|\,\bG\big]
\le \bbP\Big[ \mathcal{N}(0,1) \ge T_N\,\Big|\,\bG\Big],
\end{equation}
where
\[
T_N := \frac{\kappa - \sqrt{1-\eta_N}\,(\kappa+m(\bsigma))}{\sqrt{\eta_N}}
= \frac{\kappa\,(1-\sqrt{1-\eta_N})}{\sqrt{\eta_N}} - \sqrt{\frac{1-\eta_N}{\eta_N}}\,m(\bsigma).
\]
As $\eta_N = o(1)$, the first term is $\kappa\,\sqrt{\eta_N}/2 + o(\sqrt{\eta_N})$.
By Proposition~\ref{prop:isolation-implies-low-margin}, with probability $1-o_N(1)$ over $\bG$ we have 
$m(\bsigma)=O\big(\sqrt{\log N/N}\big)$ uniformly over all $k$-isolated $\bsigma$.
Under \eqref{eq:eta-condition}, this implies $\eta_N^{-1/2}m(\bsigma)\to 0$ uniformly.
Hence $T_N\to 0$ uniformly on an event of probability $1-o_N(1)$, and on this event the right-hand side of \eqref{eq:singleton-after-perturbation} is bounded by $1/2+o_N(1)$, which yields \eqref{eq:uniform-half}. 

For any $\btau\in \Ball_{k,\mathrm{H}}(\bsig)\setminus \{\bsig\}$, $\btau$ is not a solution. Thus there exists $1\leq a_* \leq M $ such that 
\[
\frac{\la \bg^a,\btau\ra}{\sqrt{N}}<\kappa. 
\]
Similarly conditional on $\bG$, the perturbed inner product equals
\[
\frac{\la \wt\bg^a,\btau\ra}{\sqrt{N}} = \sqrt{1-\eta_N}\frac{\la \bg^a,\btau\ra}{\sqrt{N}} +\sqrt{\eta_N} \frac{\la \bg'^a,\btau\ra}{\sqrt{N}}<\sqrt{1-\eta_N}\kappa + \sqrt{\eta_N} \frac{\la \bg'^a,\btau\ra}{\sqrt{N}}\,.
\]
Again by conditional independence, we have
\begin{equation}
    \bbP[\btau\in \cS(\wt\bG,\kappa)|\bG] <\frac12+o_N(1)\,.\notag
\end{equation}
This completes the proof.
\end{proof}

\section{Proofs of main results}\label{sec:proof}

\subsection{Auxiliary combinatorial lemmas}
We begin with two elementary balancing lemmas used in the proofs of Theorems~\ref{thm:stable-v1} and \ref{thm:stable-v2}.

\begin{lemma}
    \label{lem:partition-I}
    Consider a sequence of non-negative numbers $\{p_i\}_{1\leq i\leq M}$ that satisfies the following:
   \begin{align}
       R:=\sum_{i=1}^M p_i&\geq 0.8\,,\label{eq:sum-pi}\\
       \max_{1\leq i\leq M} p_i&\leq 0.51 \,.\label{eq:max-pi}
   \end{align}
   Then there exists a partition of $[M]=\mathsf C_1 \sqcup \mathsf C_2$, such that 
   \[
   \sum_{i\in \mathsf C_1}p_i\geq R/3 \qquad\text{  and    }\qquad \sum_{i\in \mathsf C_2}p_i\geq R/3.
   \]
   Moreover, under this partition we have 
   \begin{equation}\label{eq:partition-product}
   \Big( \sum_{i\in \mathsf C_1} p_i \Big)
   \cdot 
   \Big( \sum_{i\in \mathsf C_2} p_i \Big)
   \geq 
   \frac{R}{3} \cdot \frac{2R}{3}=\frac{2R^2}{9}\,.
   \end{equation}
\end{lemma}
\begin{proof}
    We add the items one by one until the sum first exceeds $R/3$
    \begin{equation}\label{eq:sum-first-exceed-S/3}
        \sum_{i=1}^{M_0}p_i< R/3,
        \quad 
        \sum_{i=1}^{M_0+1}p_i\geq R/3. 
    \end{equation}
    Then we consider the following two cases:
    \begin{itemize}
        \item If $\sum_{i=1}^{M_0+1}p_i\leq 2R/3$, then $\sum_{i=M_0+2}^M p_i\geq R/3$. We can take $\mathsf C_1=[M_0+1]$ and $\mathsf C_2=[M]\setminus \mathsf C_1$.
        \item If $\sum_{i=1}^{M_0+1}p_i\geq 2R/3$, then by \eqref{eq:sum-first-exceed-S/3} we have $p_{M_0+1}\geq R/3$. Combining with \eqref{eq:max-pi}, we can take $\mathsf C_1=\{M_0+1\}$ and $\mathsf C_2=[M]\setminus \mathsf C_1$.
        (Note that $0.51\leq 2R/3$ since $R\geq 0.8$.)
    \end{itemize}
    The first claim follows. The second claim \eqref{eq:partition-product} is from $t(1-t)\geq 2/9$ for $1/3\leq t\leq 2/3$.
    The proof is complete.
\end{proof}

\begin{lemma}
\label{lem:partition-II}
    Consider a sequence of non-negative numbers $\{p_i\}_{1\leq i\leq M}$ that satisfies 
    \begin{align*}
        \sum_{i=1}^M p_i&\geq R\,,\\
        \max_{1\leq i\leq M} p_i &\leq \frac{R}{100}\,.
    \end{align*}
   Then there exists a partition $[M]=\mathsf C_1\sqcup \mathsf C_2$ such that 
   \[
   \sum_{i\in \mathsf C_1}p_i\geq 0.49R \qquad\text{  and    }\qquad \sum_{i\in \mathsf C_2}p_i\geq 0.49R.
   \]
\end{lemma}

\begin{proof}
    We add the items one by one until the sum first exceeds $0.49 R$; thus 
    \begin{equation}
        \sum_{i=1}^{M_0}p_i< 0.49R,\quad\sum_{i=1}^{M_0+1}p_i\geq 0.49 R. \notag
    \end{equation}
    Noting that $p_i\leq R/100$, we have
    $\sum_{i=1}^{M_0+1}p_i\leq 0.5 R$. This implies $\sum_{i=M_0+2}^M p_i\geq 0.49R$. Then $\mathsf C_1=[M_0+1]$ and $\mathsf C_2=[M]\setminus \mathsf C_1$ satisfy the desired condition.
\end{proof}

\subsection{Proof of main obstruction theorem}
\label{subsec:proof-main-obstruction-theorem}
We now prove Theorem~\ref{thm:stable-v1}. We omit the proof of Theorem~\ref{thm:k-isolation}, which follows in exactly the same way after replacing the scale $\iota N$ by $k_N$.

\begin{proof}[Proof of Theorem~\ref{thm:stable-v1}]
Let $p_N:=\bbP[\mathcal A_N\text{ locates an $\iota N$-isolated solution of }(\bG,\kappa)]$.
Below we assume without loss of generality that $p_N\geq 0.801$ as otherwise there is nothing to prove.
We couple $(\bG,\widetilde \bG)$ as in \eqref{eq:noise} at noise level $\eta_N$.
Write the success event as (for any $N\times M$ matrix $\bX$)
\[
\mathrm{Iso}(\bX):=\Big\{\mathrm{dist}(\mathcal A_N(\bX),\cS^\circ(\bX,\kappa,\iota))\le \frac{\sqrt{\iota N}}{3}\Big\}.
\]

\smallskip
\noindent\emph{Step 1: choosing a ``good'' instance for which both $\bG$ succeeds and $\widetilde \bG$ succeeds conditionally with probability $\approx p_N$.}
Represent the coupling via a shared Gaussian component:
\[
\bG=\sqrt{1-\lambda_N}\,\bG_1+\sqrt{\lambda_N}\,\bG_2,\qquad 
\widetilde \bG=\sqrt{1-\lambda_N}\,\bG_1+\sqrt{\lambda_N}\,\bG_3,
\]
for i.i.d.\ standard Gaussian matrices $\bG_1,\bG_2,\bG_3$ and $\lambda_N=1-\sqrt{1-\eta_N}$ (recall \eqref{eq:noise}).
Conditioned on $(\bG_1,\omega)$, the randomness used to generate $\bG$ and $\widetilde \bG$ is independent, so by Jensen's inequality,
\begin{equation}
\label{eq:jensen-v1}
\bbP[\mathrm{Iso}(\bG)\cap \mathrm{Iso}(\widetilde \bG)]
=\bbE\Big[\bbP(\mathrm{Iso}(\bG)\mid \bG_1,\omega)^2\Big]
\ge p_N^2,
\end{equation}
where $\omega$ denotes the internal randomness of $\mathcal A_N$ (shared between the two runs, as in the stability definition).
Therefore,
\begin{equation}
\label{eq:cond-iso-lb}
\bbE\big[\bbP(\mathrm{Iso}(\widetilde \bG)\mid \bG,\omega)\,\big|\,\mathrm{Iso}(\bG)\big]
=\frac{\bbP[\mathrm{Iso}(\bG)\cap \mathrm{Iso}(\widetilde \bG)]}{\bbP[\mathrm{Iso}(\bG)]}
\ge p_N.
\end{equation}
Let $\mathcal G_{\mathrm{margin}}$ be the high-probability event from Corollary~\ref{cor:noise-destroys} (with $k=\lfloor \iota N\rfloor$), and let $\mathcal G_{\mathrm{stab}}$ be the event that $\|\mathcal A_N(\bG)-\mathcal A_N(\widetilde \bG)\|_2\le \rho_N$.
We have $\bbP[\mathcal G_{\mathrm{margin}}]=1-o_N(1)$ and $\bbP[\mathcal G_{\mathrm{stab}}]\ge 1-t_N$. Thus, intersecting with $\mathrm{Iso}(\bG)$, we may choose a realization $(\bG,\omega)$ such that:
\begin{enumerate}[label=(\roman*), leftmargin=2.5em]
\item $\mathrm{Iso}(\bG)$ holds;
\item $\mathcal G_{\mathrm{margin}}$ holds;
\item $\bbP(\mathrm{Iso}(\widetilde \bG)\mid \bG,\omega)\ge p_N-o_N(1)$;
\item $\bbP(\mathcal G_{\mathrm{stab}}\mid \bG,\omega)\ge 1-o_N(1)$.
\end{enumerate}
The existence of such $(\bG,\omega)$ is guaranteed by considering the conditional distribution $\bbP(\cdot\,|\,\mathrm{Iso}(\bG))$, under which $\mathcal{G}_{\mathrm{margin}}$ and $\mathcal{G}_{\mathrm{stab}}$ still hold with probability $1-o_N(1)$. For (iii), by \eqref{eq:cond-iso-lb}, we have
\[
\bbP\big(  \bbP(\mathrm{Iso}(\widetilde \bG)\mid \bG,\omega)\ge p_N-\varepsilon_* \,|\,\mathrm{Iso}(\bG) \big) \geq \frac{\bbE[ \mathrm{Iso}(\wt\bG)|\mathrm{Iso}(\bG)] - p_N+\varepsilon_* }{1-p_N+\varepsilon_*}\geq \frac{\varepsilon_*}{1+\varepsilon_*}\,.
\]
The desired existence claim then follows by choosing $\varepsilon_*$ appropriately.

Fix such $(\bG,\omega)$ for the remainder of the proof. Let $\bsigma_\star\in\cS^\circ(\bG,\kappa,\iota)$ be a closest isolated solution, so that
\begin{equation}
\label{eq:sigma-star-close}
\|\mathcal A_N(\bG)-\bsigma_\star\|_2\le \frac{\sqrt{\iota N}}{3}.
\end{equation}

\smallskip
\noindent\emph{Step 2: defining a discrete search region in $\bSigma_N$ under the perturbed instance.}
Define
\[
\mathsf{CAND}_\kappa:=\{\btau\in\bSigma_N:\ \|\btau-\mathcal A_N(\bG)\|_2 \le \frac{\sqrt{\iota N}}{3}+\rho_N\}.
\]
On the event $\mathcal G_{\mathrm{stab}}\cap \mathrm{Iso}(\widetilde \bG)$, there exists $\widetilde\bsigma\in\cS^\circ(\widetilde \bG,\kappa,\iota)$ with
\[
\|\widetilde\bsigma-\mathcal A_N(\widetilde \bG)\|_2\le \frac{\sqrt{\iota N}}{3},
\qquad
\|\mathcal A_N(\widetilde \bG)-\mathcal A_N(\bG)\|_2\le \rho_N,
\]
hence $\widetilde\bsigma\in \mathsf{CAND}_\kappa$.
Moreover, since $\rho_N=o(\sqrt N)$, 
its Hamming diameter is $\mathrm{diam}_{\dH}(\mathsf{CAND}_\kappa)<\iota N$.
Thus, if $\mathsf{CAND}_\kappa$ contains any $\iota N$-isolated solution of $(\wt\bG,\kappa)$, it contains \emph{at most one} solution of $\cS(\widetilde \bG,\kappa)$.

Let
\[
S:=\bbP\big[|\cS(\widetilde \bG,\kappa)\cap \mathsf{CAND}_\kappa|=1\,\big|\,\bG,\omega\big].
\]
From item (iii) and the preceding discussion, we have
\[
S \ge \bbP(\mathrm{Iso}(\widetilde \bG)\cap \mathcal G_{\mathrm{stab}}\mid \bG,\omega)
\ge p_N - o_N(1).
\]
For each $\btau\in \mathsf{CAND}_\kappa$, define
\[
p_\btau := \bbP\big[\btau\in\cS(\widetilde \bG,\kappa)\ \text{and}\ |\cS(\widetilde \bG,\kappa)\cap \mathsf{CAND}_\kappa|=1\,\big|\,\bG,\omega\big].
\]
The events defining $p_\btau$ are disjoint across $\btau\in \mathsf{CAND}_\kappa$, so
\begin{equation}
\label{eq:sum-p-1}
\sum_{\btau\in \mathsf{CAND}_\kappa} p_\btau = S.
\end{equation}

Finally, we claim $\max_{\btau\in \mathsf{CAND}_\kappa} p_\btau\le \tfrac12+o_N(1)$.
Indeed, $p_\btau\le \bbP[\btau\in\cS(\widetilde \bG,\kappa)\mid \bG]$.
Also, by \eqref{eq:sigma-star-close} and the definition of $\mathsf{CAND}_\kappa$,
for any $\btau\in \mathsf{CAND}_\kappa$ we have
\[
\|\btau-\bsigma_\star\|_2 \le \|\btau-\mathcal A_N(\bG)\|_2+\|\mathcal A_N(\bG)-\bsigma_\star\|_2 \le \frac{2\sqrt{\iota N}}{3}+\rho_N,
\]
hence
\[
\dH(\btau,\bsigma_\star)=\frac{\|\btau-\bsigma_\star\|_2^2}{4}\le \Big(\frac{2}{3}+o_N(1)\Big)^2\frac{\iota N}{4}< \iota N
\]
for large $N$.
Therefore $\mathsf{CAND}_\kappa\subseteq \Ball_{\iota N,\mathrm{H}}(\bsigma_\star)$, and Corollary~\ref{cor:noise-destroys} implies
\begin{equation}
\label{eq:max-p-1}
\max_{\btau\in \mathsf{CAND}_\kappa} p_\btau \le \tfrac12+o_N(1).
\end{equation}

\smallskip
\noindent\emph{Step 3: partitioning the search region and deriving the contradiction.}
Apply Lemma~\ref{lem:partition-I} to the numbers $\{p_\btau:\btau\in \mathsf{CAND}_\kappa\}$.
Using \eqref{eq:sum-p-1}--\eqref{eq:max-p-1}, we obtain a partition $\mathsf{CAND}_\kappa=\mathsf C_1\sqcup \mathsf C_2$ such that
\begin{equation}
\label{eq:partition-sums}
\begin{split}
    & \sum_{\btau\in \mathsf{C}_i} p_\btau \ge \frac{S}{3},\qquad i=1,2.\\
    & \Big(\sum_{\btau\in \mathsf{C}_1} p_\btau \Big)
    \cdot 
    \Big(\sum_{\btau\in \mathsf{C}_2} p_\btau \Big) \geq \frac{2S^2}{9}.
\end{split}
\end{equation}
Define the events
\[
E_i := \{\cS(\widetilde \bG,\kappa)\cap \mathsf{C}_i\ne\emptyset\},\qquad i=1,2.
\]
Then $E_i$ contains the disjoint union of events counted by $p_\btau$ for $\btau\in \mathsf{C}_i$, so
\[
\bbP[E_i\mid \bG,\omega]\ge \sum_{\btau\in \mathsf{C}_i} p_\btau .
\]

We now apply Proposition~\ref{prop:pitt} conditionally on $(\bG,\omega)$.
Consider the Gaussian vector consisting of all constraint values
\[
X_{\tau,a} := \frac{\la \widetilde \bg^a,\btau\ra}{\sqrt{N}},\qquad \btau\in \mathsf{CAND}_\kappa,\ 1\le a\le M.
\]
Conditional on $\bG$, the random part of $X_{\btau,a}$ comes from $\bG'$, and one checks that
\[
\mathrm{Cov}(X_{\btau,a},X_{\btau',a'}) = \eta_N\,q(\btau,\btau')\,\ind\{a=a'\},
\qquad q(\btau,\btau'):=\frac1N\sum_{i=1}^N \btau_i\btau'_i = 1-\frac{2\dH(\btau,\btau')}{N}.
\]
Since $\mathsf{CAND}_\kappa$ has Hamming diameter smaller than $\iota N/3$, for large $N$ we have $q(\btau,\btau')\ge 0$ for all $\btau,\btau'\in \mathsf{CAND}_\kappa$.
Moreover, each $E_i$ is coordinate-wise increasing in the collection $(X_{\btau,a})$.
Therefore, by Pitt's inequality (recall \eqref{eq:partition-sums}),
\[
\bbP[E_1\cap E_2\mid \bG,\omega] \ge \bbP[E_1\mid \bG,\omega]\,\bbP[E_2\mid \bG,\omega]
\ge \frac{2S^2}{9}.
\]
On the other hand, $E_1\cap E_2$ implies that $\cS(\widetilde \bG,\kappa)\cap \mathsf{CAND}_\kappa$ contains at least two solutions, hence
\[
\bbP[E_1\cap E_2\mid \bG,\omega] \le \bbP\big[|\cS(\widetilde \bG,\kappa)\cap \mathsf{CAND}_\kappa|\ge 2\,\big|\,\bG,\omega\big]\le 1-S.
\]
Combining yields
\[
1-S \ge \frac{2S^2}{9}.
\]
Since $S\ge p_N-o_N(1)$, this implies
\[
1-p_N \ge \frac{2p_N^2}{9}+o_N(1),
\]
which forces $p_N\le \frac{3\sqrt{17}-9}{4}+o_N(1)$.
\end{proof}

\subsection{High-success algorithms}
Here we prove Theorem~\ref{thm:stable-v2}, which uses a similar idea. 
\begin{proof}[Proof of Theorem~\ref{thm:stable-v2}]
    By Markov's inequality, with probability at least $1-\delta^{1/2}$ over the randomness of $(\bG,\omega)$, 
    \begin{equation}\label{eq-prob-locate-solution}
    \mathbb P[ \mathcal{A}\text{ $\iota N$-locates a solution of $\wt\bG$}\,|\, \bG,\omega  ] \geq 1-\delta^{1/2}\,.
    \end{equation}
    We denote the above event by $\mathrm{Loc}(\wt\bG)$.
    Recall that $p_N:=\bbP(\mathcal{A}_N\text{ locates an $\iota N$-isolated solution of $(\bG,\kappa)$})$. Similarly to \eqref{eq:jensen-v1}, we have
    \begin{equation}
        \bbP\big(\mathrm{Iso}(\bG)\cap\mathrm{Iso}(\wt\bG)\big) \geq p_N^2.\notag
    \end{equation}
    This implies with probability at least $p_N^2/2$ over $(\bG,\omega)$,
    \begin{equation}\label{eq:prob-locate-singleton}
    \bbP\big( \mathrm{Iso}(\bG)\cap\mathrm{Iso}(\wt\bG)\,|\,\bG,\omega \big) \geq p_N^2/2\,.
    \end{equation}
    If $p_N$ satisfies $p_N^2/2+1-\delta^{1/2} > 1$ (for instance, it suffices to take $p_N\geq 2\delta^{1/4}$), we can choose a realization $(\bG,\omega)$ such that:
    \begin{enumerate}[label=(\roman*), leftmargin=2.5em]
        \item $\mathcal{G}_{\mathrm{margin}}$ holds,
        \item $\mathbb P(\mathrm{Loc}(\wt\bG)\,|\,\bG,\omega)\geq 1-\delta^{1/2}$,
        \item $\mathbb P(\mathcal{G}_{\mathrm{stab}}\,|\,\bG,\omega)\geq 1-o_N(1)$,
        \item $\mathbb P(\mathrm{Iso}(\wt\bG)\cap \mathrm{Iso}(\bG)\,|\,\bG,\omega)\geq p_N^2/2$.
    \end{enumerate}
    We fix such $(\bG,\omega)$ for the remainder of the proof. Let $\bsigma_\star\in \cS^\circ (\bG,\kappa,\iota)$ be a closest isolated solution, so that
    \begin{equation}
    \label{eq:sigma-star-close-2}
    \|\mathcal A_N(\bG)-\bsigma_\star\|_2\leq \frac{\sqrt{\iota N}}{3}\,.
    \end{equation}
    We define the ball around $\mathcal{A}_N(\bG)$ as
    \[
    \mathsf{CAND}_\kappa:=\{\btau\in \bSigma_N:\,\|\btau-\mathcal A_N(\bG)\|_2\leq \frac{\sqrt{\iota N}}{3}+\rho_N\}\,.
    \]
    On the event $\mathcal G_{\mathrm{stab}}\cap \mathrm{Iso}(\widetilde \bG)$, there exists $\widetilde\bsigma\in\cS^\circ(\widetilde \bG,\kappa,\iota)$ with
\[
\|\widetilde\bsigma-\mathcal A_N(\widetilde \bG)\|_2\le \frac{\sqrt{\iota N}}{3},
\qquad
\|\mathcal A_N(\widetilde \bG)-\mathcal A_N(\bG)\|_2\le \rho_N,
\]
hence $\widetilde\bsigma\in \mathsf{CAND}_\kappa$.
Moreover, since $\rho_N=o(\sqrt N)$, 
we have $\mathrm{diam}_{\dH}(\mathsf{CAND}_\kappa)<\iota N$.
Thus, if $\mathsf{CAND}_\kappa$ contains any $\iota N$-isolated solution, it contains \emph{at most one} solution of $\cS(\widetilde \bG,\kappa)$.

Let
\[
S:=\bbP\big[|\cS(\widetilde \bG,\kappa)\cap \mathsf{CAND}_\kappa|=1\,\big|\,\bG,\omega\big].
\]
From Items (iii),(iv) and the preceding implication, we have
\begin{equation}\label{eq:lower-boundS}
S \ge \bbP(\mathrm{Iso}(\widetilde \bG)\cap \mathcal G_{\mathrm{stab}}\mid \bG,\omega)
\ge p_N^2/2 - o_N(1).
\end{equation}
For each $\btau\in \mathsf{CAND}_\kappa$, define
\[
p_\btau := \bbP\big[\btau\in\cS(\widetilde \bG,\kappa)\ \text{and}\ |\cS(\widetilde \bG,\kappa)\cap \mathsf{CAND}_\kappa|=1\,\big|\,\bG,\omega\big].
\]
The events defining $p_\btau$ are disjoint across $\btau\in \mathsf{CAND}_\kappa$, and recall \eqref{eq:lower-boundS}. We have
\begin{equation}
\label{eq:sum-p}
\sum_{\btau\in \mathsf{CAND}_\kappa} p_\btau = S \geq p_N^2/2 - o_N(1).
\end{equation}

Finally, we claim $\max_{\btau\in \mathsf{CAND}_\kappa} p_\btau\le \tfrac12+o_N(1)$.
Indeed, $p_\btau\le \bbP[\btau\in\cS(\widetilde \bG,\kappa)\mid \bG]$.
Also, by \eqref{eq:sigma-star-close-2} and the definition of $\mathsf{CAND}_\kappa$,
for any $\btau\in \mathsf{CAND}_\kappa$ we have
\[
\|\btau-\bsigma_\star\|_2 \le \|\btau-\mathcal A_N(\bG)\|_2+\|\mathcal A_N(\bG)-\bsigma_\star\|_2 \le \frac{2\sqrt{\iota N}}{3}+\rho_N,
\]
implying that $\dH(\btau,\bsigma_\star)\leq \iota N$ for large $N$.
Therefore $\mathsf{CAND}_\kappa\subseteq \Ball_{\iota N,\mathrm{H}}(\bsigma_\star)$, and Corollary~\ref{cor:noise-destroys} implies
\begin{equation}
\label{eq:max-p}
\max_{\btau\in \mathsf{CAND}_\kappa} p_\btau \le \tfrac12+o_N(1).
\end{equation}
Now we claim: 
\begin{equation}\label{eq:claim-prob-no-solution}
    \bbP(\mathsf{CAND}_\kappa \text{ contains no solution of $\wt\bG$}\,|\,\bG,\omega) \geq \max\Big\{\frac{p_N^2}{500}, \frac{p_N^4}{40}\Big\}.
\end{equation}
To show \eqref{eq:claim-prob-no-solution} we partition the set $\mathsf{CAND}_\kappa$ into disjoint parts $\mathsf C_1,\mathsf C_2$ such that each part contains no solution with certain probability and then apply Pitt's inequality. 

We consider two cases. First, suppose there exists a $\btau\in \mathsf{CAND}_\kappa$ such that $p_\btau \geq p_N^2/200$. Then we can pick $\mathsf C_1=\{\btau\}$ and $\mathsf C_2=\mathsf{CAND}_\kappa\setminus \mathsf C_1$.
Noting that $\sum_{\btau\in \mathsf C_i}p_\btau$ lower-bounds the probability that the other side contains no solution, not the same side, we obtain:
\begin{align*}
    &\mathbb P(\mathsf C_1\text{ contains no solution}\,|\,\bG,\omega) = \mathbb P(\btau\notin \cS(\wt\bG,\kappa)\,|\,\bG,\omega) \geq 0.49,\\
    &\mathbb P(\mathsf C_2\text{ contains no solution}\,|\,\bG,\omega) \geq p_\btau \geq p_N^2/200,
\end{align*}
where the first inequality follows from the definition of $\cG_{\mathrm{margin}}$. The first item of right-hand side of \eqref{eq:claim-prob-no-solution} follows by applying Pitt's inequality (Proposition~\ref{prop:pitt}). 

Otherwise if $p_\btau \leq \frac{p_N^2}{200}$ for all $\btau\in \mathsf{CAND}_\kappa$, by Lemma~\ref{lem:partition-II}, there exists a partition $\mathsf{CAND}_\kappa=\mathsf C_1\sqcup \mathsf C_2$ with
\begin{equation*}
    \sum_{\btau\in \mathsf{C}_i} p_\btau \geq \frac{p_N^2}{5}\,,\text{ for }i=1,2.
\end{equation*}
This means that
\[
\mathbb P(\mathsf{C}_i\text{ contains no solution}\,|\,\bG,\omega) \geq \frac{p_N^2}{5}\,\text{ for } i=1,2.
\]
Then the second item of \eqref{eq:claim-prob-no-solution} is from Pitt's inequality. Recalling the definition of $\mathrm{Loc}(\wt\bG)$, and using the stability event \(\mathcal{G}_{\mathrm{stab}}\), we have
\begin{equation}
\label{eq:claim-prob-no-solution-upperbound}
    \bbP(\mathsf{CAND}_\kappa \text{ contains no solution of $\wt\bG$}\,|\,\bG,\omega)\leq \delta^{1/2}+o_N(1)\,.
\end{equation}
Combining \eqref{eq:claim-prob-no-solution} and \eqref{eq:claim-prob-no-solution-upperbound} yields
\[
p_N\leq \min\{ \sqrt{500}\delta^{\frac{1}{4}}, 40^{\frac{1}{4}}\delta^{\frac{1}{8}} \} + o_N(1).
\]
The proof is complete.
\end{proof}

\subsection{Extensions to general perceptrons}
We now prove the corresponding statements for perceptron models defined by a finite union of intervals.

Before we start the proof, we first define \emph{margin} for general perceptrons. For simplicity, we write $\mathcal{U}=\sqcup_{i=1}^L \mathcal{I}_i$, where we denote $\mathcal{I}_i=[a_i,b_i]$ and the $a_i$'s and $b_i$'s satisfy $a_1<b_1<a_2<b_2<\dots<
a_L<b_L$. 
For each constraint $a$, define its margin under $\bsigma$ by
\[
m_a(\bsigma):=\mathrm{dist}\!\left( \frac{\langle \bg^a,\bsigma\rangle}{\sqrt{N}},\; \mathbb{R}\setminus \mathcal{U}\right).
\]
Equivalently, if $ \langle \bg^a, \bsigma \rangle/\sqrt{N} \in [a_i,b_i]$ for some $i$, then
\[
m_a(\bsigma)=\min\left\{ \frac{\langle \bg^a,\bsigma\rangle}{\sqrt{N}}-a_i,\; b_i- \frac{\langle \bg^a,\bsigma\rangle}{\sqrt{N}}\right\}.
\]
The margin of the solution $\bsigma$ is then defined by
\begin{equation}\label{eq:margin-SBP}
    m(\bsigma):=\min_{1\leq a \leq M} m_a(\bsigma).
\end{equation}
Similarly to the proof of Proposition~\ref{prop:isolation-implies-low-margin}, we have that with probability $1-o_N(1)$, 
\begin{equation}\label{eq:margin-singleton-SBP}
    m(\bsig) = O\left(\sqrt{\frac{\log N}{N}}\right) \text{ uniformly over all }\bsig\in \bigcup_{k\geq 1}\mathcal{S}^\circ_k(\bG,\mathcal{U}).
\end{equation}
Still, we consider $(\bG,\wt\bG)$ as in \eqref{eq:noise}. Similarly to Corollary~\ref{cor:noise-destroys}, we can show that with probability $1-o_N(1)$ over $\bG$, the following holds: for every $\bsig\in\mathcal{S}^\circ_k(\bG,\mathcal{U})$ and every $\btau\in \mathsf B_{k,H}(\bsig)$,
\begin{equation}\label{eq:wtbG-sol-upper-bound-SBP}
    \mathbb P [ \btau\in\mathcal{S}(\wt\bG,\mathcal{U}) \,|\, \bG ] \leq \frac{1}{2}+o_N(1)\,.
\end{equation}

\begin{proof}[Proof of Theorem~\ref{thm:stable-SBP}]
Let \[p_{N,\mathcal{U}}:=\mathbb P[\mathcal{A}_N \text{ $k_N$-locates a $k_N$-isolated solution of }(\bG,\mathcal{U})].\] 
We couple $(\bG,\wt\bG)$ as in \eqref{eq:noise} and let
\[
\mathrm{Iso}(\bH,\mathcal{U}):=\Big\{\mathrm{dist}(\mathcal{A}_N(\bH),\mathcal{S}^\circ_{k_N}(\bH,\mathcal{U}))\leq \frac{\sqrt{k_N}}{3}\Big\}.
\]
Similarly, we let $\mathcal{G}_{\mathrm{stab},\mathcal{U}}$ be the event that $\|\mathcal{A}_N(\bG)-\mathcal{A}_N(\wt\bG)\|_2\leq \rho_N$, and let $\mathcal{G}_{\mathrm{margin},\mathcal{U}}$ be the event in \eqref{eq:wtbG-sol-upper-bound-SBP}, which is measurable w.r.t. $\bG$. Steps 1-2 in the proof of Theorem~\ref{thm:stable-v1} carry over verbatim here. Thus, we can choose $(\bG,\omega)$ where $\omega$ is the internal randomness of $\mathcal{A}_N$ such that:
\begin{itemize}
    \item $\mathrm{Iso}(\bG,\mathcal{U})$ holds;
    \item $\mathcal{G}_{\mathrm{margin},\mathcal{U}}$ holds;
    \item $\mathbb P[\mathrm{Iso}(\wt\bG,\mathcal{U})\,|\,\bG,\omega]\geq p_{N,\mathcal{U}}-o_N(1)$;
    \item $\mathbb P[ \mathcal{G}_{\mathrm{stab},\mathcal{U}}\,|\,\bG,\omega ]\geq 1-o_N(1)$;
    \item $\|\bG\|_\infty\leq 10\sqrt{\log N}$.
\end{itemize}
We fix such $(\bG,\omega)$ for the remainder of the proof. Let $\bsig_* \in \mathcal{S}^\circ_{k_N}(\bG,\mathcal{U})$ be the closest solution of $\bG$ that $\mathcal{A}_N(\bG)$ locates, i.e.
\begin{equation}
    \|\mathcal{A}_N(\bG)-\bsig_*\|_2\leq \frac{\sqrt{k_N}}{3}.
\end{equation}
We define 
\[
\mathsf{CAND}_{\mathcal U}:=\{ \btau\in \bSig_N:\,\|\btau-\mathcal{A}_N(\bG)\|_2\leq \frac{\sqrt{k_N}}{3}+\rho_N \}.
\]
Let $S_\mathcal{U}:=\mathbb P[|\mathcal{S}(\wt\bG,\mathcal U)\cap \mathsf{CAND}_{\mathcal U}|=1\,|\,\bG,\omega]$. For each $\btau\in \mathsf{CAND}_{\mathcal U}$, define
\[
p_{\btau,\mathcal{U}}:=\mathbb P[\btau\in\mathcal{S}(\wt\bG,\mathcal{U})\text{ and }|\mathcal{S}(\wt\bG,\mathcal U)\cap \mathsf{CAND}_{\mathcal U}|=1 \, | \, \bG,\omega]\,.
\]
Similarly as in step~2, we have
\begin{equation}
\begin{split}
   & \sum_{\btau\in \mathsf{CAND}_{\mathcal U}} p_{\btau,\mathcal{U}} =S_\mathcal{U}\geq p_{N,\mathcal{U}}-o_N(1)\,.\\
   & \max_{\btau\in \mathsf{CAND}_{\mathcal U}}p_{\btau,\mathcal{U}}\leq \frac{1}{2}+o_N(1).
\end{split}
\end{equation}
Again by Lemma~\ref{lem:partition-I}, there exists a partition $\mathsf{CAND}_{\mathcal U}=\mathsf{C}_{1,\mathcal{U}}\sqcup \mathsf{C}_{2,\mathcal{U}}$ such that
\begin{equation}\label{eq:partition-sum-SBP}
\begin{split}
    & \sum_{\btau\in \mathsf{C}_{i,\mathcal{U}}} p_{\btau,\mathcal{U}} \geq \frac{S_\mathcal{U}}{3},\qquad i=1,2\,,\\
    & \left( \sum_{\btau\in \mathsf{C}_{1,\mathcal{U}}} p_{\btau,\mathcal{U}} \right) \left( \sum_{\btau\in \mathsf{C}_{2,\mathcal{U}}} p_{\btau,\mathcal{U}} \right) \geq \frac{2S_\mathcal{U}^2}{9}\,.
\end{split}
\end{equation}
Defining the events $E_{i,\mathcal{U}}=\{\mathcal{S}(\wt\bG,\mathcal{U})\cap \mathsf{C}_{i,\mathcal{U}}\neq \emptyset\}$, we have
\begin{equation}\label{eq:def-Eiu}
    \mathbb P[E_{i,\mathcal{U}}\,|\,\bG,\omega]\geq \sum_{\btau\in \mathsf{C}_{i,\mathcal{U}}} p_{\btau,\mathcal{U}}\qquad i=1,2.
\end{equation}
We cannot apply Pitt's inequality directly to the events $E_{i,\mathcal{U}}$, since in contrast to the one-sided case, the constraints are no longer monotone. Nevertheless, because the diameter of $\mathsf{CAND}_{\mathcal U}$ is small, only one side of each constraint is effectively active, allowing us to recover the monotonicity required for Pitt's inequality. We make this precise below.

Recall that $\bsig_*\in \mathcal{S}^\circ_{k_N}(\bG,\mathcal{U})$ is the $k_N$-isolated solution located by $\mathcal{A}_N(\bG)$ and that $\bsig_*\in \mathsf{CAND}_{\mathcal U}$. We define the following sets of \emph{relevant} constraints that are close to the boundary
\begin{align}
    & \mathsf{Rel}^-:=\left\{1\leq c\leq M:\,m_c(\bsig_*)\leq \frac{5}{\log N}\,, \exists  1\leq i_c\leq L, \frac{\langle \bg^c,\bsig_*\rangle}{\sqrt{N}}-a_{i_c}\leq \frac{5}{\log N}\right\}\,,\\
    & \mathsf{Rel}^+:=\left\{1\leq c\leq M:\,m_c(\bsig_*)\leq \frac{5}{\log N}\,, \exists  1\leq i_c\leq L, b_{i_c}-\frac{\langle \bg^c,\bsig_*\rangle}{\sqrt{N}}\leq \frac{5}{\log N}\right\}\,.
\end{align}

Since $\|\bG\|_\infty \leq 10\sqrt{\log N}$ and $\mathrm{diam}_{d_\mathrm H}(\mathsf{CAND}_{\mathcal U})\leq \frac{\sqrt{N}}{(\log N)^2}$, we have that for all $\btau\in \mathsf{CAND}_{\mathcal U}$, 
\begin{equation}\label{eq:condition-bG-eff-side}
\begin{split}
    & |\frac{\langle \bg^c,\btau\rangle}{\sqrt{N}}-a_{i_c}|\leq 6/\log N \text{ for } c\in\mathsf{Rel}^-\,,\\
    &|\frac{\langle \bg^c,\btau\rangle}{\sqrt{N}}-b_{i_c}|\leq 6/\log N \text{ for } c\in\mathsf{Rel}^+.
\end{split}
\end{equation}
This means all $\btau\in \mathsf{CAND}_{\mathcal U}$ are targeting the same boundary. 
Recall the definition of $E_{i,\mathcal{U}}$ above \eqref{eq:def-Eiu}, we define $\wt E_{i,\mathcal{U}}$ to be the event that there exists a $\btau\in \mathsf{C}_{i,\mathcal{U}}$, such that 
\begin{equation}
\begin{split}
    & \frac{\langle \wt\bg^c, \btau\rangle}{\sqrt{N}} \geq a_{i_c} \text{ for }c\in \mathsf{Rel}^{-}\,,\\
    & \frac{\langle \wt\bg^c, \btau\rangle}{\sqrt{N}} \leq b_{i_c} \text{ for }c\in \mathsf{Rel}^{+}.
\end{split}
\end{equation}
Note that $\wt\bg^c=\sqrt{1-\eta_N}\bg^c + \sqrt{\eta_N}\bg'^c$. For each $1\leq c\leq M$ and $\btau\in \mathsf{CAND}_{\mathcal U}$, we have
\begin{equation}
\mathbb P\left[ \sqrt{\eta_N}\frac{|\langle \bg'^c,\btau\rangle|}{\sqrt{N}} \geq \frac{1}{\log N} \right] \leq \exp(-N^{1-o(1)}).
\end{equation}
In addition, $|\mathsf{CAND}_{\mathcal U}\cap \bSig_N|=\exp(N^{1/2-o(1)})$. This implies 
\begin{equation}
    \mathbb P\left[ \sqrt{\eta_N}\frac{|\langle \bg'^c,\btau\rangle|}{\sqrt{N}} \leq \frac{1}{\log N}\text{ for all }\btau\in \mathsf{CAND}_{\mathcal U} \text{ and }1\leq c\leq M\right] = 1-o(1)\,.
\end{equation}
We denote this event by $\mathcal{G}'$. Recall that we have conditioned on $\bG$ and it satisfies \eqref{eq:condition-bG-eff-side}. Under $\mathcal{G'}\cap \wt E_{i,\mathcal{U}}$, we have that there exists a $\btau\in \mathsf{CAND}_{\mathcal U}$ such that 
\begin{equation*}
    \begin{split}
        &  a_{i_c} \leq \frac{\langle \wt\bg^c, \btau\rangle}{\sqrt{N}}\leq a_{i_c}+\frac{7}{\log N}<b_{i_c}  \text{ for }c\in \mathsf{Rel}^{-}\,,\\
    & a_{i_c}<b_{i_c}-\frac{7}{\log N}<\frac{\langle \wt\bg^c, \btau\rangle}{\sqrt{N}} \leq b_{i_c} \text{ for }c\in \mathsf{Rel}^{+}\,,\\
    & \exists 1\leq i_c\leq L, a_{i_c}<\frac{\langle \wt\bg^c, \btau\rangle}{\sqrt{N}}<b_{i_c} \text{ for }c\in [M]\setminus (\mathsf{Rel}^{-}\cup\mathsf{Rel}^{+}).
    \end{split}
\end{equation*}
 Therefore, under such $\bG$ we obtain $\mathcal{G'}\cap \wt E_{i,\mathcal{U}} \subset E_{i,\mathcal{U}}$ for $i=1,2$. Moreover, we have $\mathcal{G'}\cap  E_{i,\mathcal{U}} \subset \wt E_{i,\mathcal{U}}$ by \eqref{eq:condition-bG-eff-side}. This further implies $\mathcal{G'}\cap \wt E_{i,\mathcal{U}} = E_{i,\mathcal{U}} \cap \mathcal{G}'$.

Since the collections $\{\bg'^c\}_c$ are independent across different $c$, we may apply Pitt's inequality to the variables $\langle \bg'^c,\btau\rangle$ for $c\in\mathsf{Rel}^-$ and $\langle -\bg'^c,\btau\rangle$ for $c\in\mathsf{Rel}^+$. It follows that
\begin{equation}\label{eq:intermediate-FKG-SBP}
    \mathbb P[ \wt E_{1,\mathcal{U}}\cap  \wt E_{2,\mathcal{U}}\,|\,\bG,\omega] \geq \mathbb P[ \wt E_{1,\mathcal{U}}\,|\,\bG,\omega]\mathbb P[ \wt E_{2,\mathcal{U}}\,|\,\bG,\omega].
\end{equation}
Note that 
\begin{equation}\label{eq:intermediate-FKG-SBP-1}
    \mathbb P[ \wt E_{i,\mathcal{U}}\,|\,\bG,\omega] \geq \mathbb P[ \wt E_{i,\mathcal{U}}\cap \mathcal{G}'\,|\,\bG,\omega] = \mathbb P[  E_{i,\mathcal{U}}\cap \mathcal{G}'\,|\,\bG,\omega] \geq \mathbb P[  E_{i,\mathcal{U}}\,|\,\bG,\omega] - o_N(1).
\end{equation}
and 
\begin{equation}\label{eq:intermediate-FKG-SBP-2}
    \mathbb P[ \wt E_{1,\mathcal{U}}\cap  \wt E_{2,\mathcal{U}}\,|\,\bG,\omega]\leq \mathbb P[ \wt E_{1,\mathcal{U}}\cap  \wt E_{2,\mathcal{U}}\cap \mathcal{G}'\,|\,\bG,\omega] + o_N(1)\leq \mathbb P[  E_{1,\mathcal{U}}\cap   E_{2,\mathcal{U}}\,|\,\bG,\omega] + o_N(1).
\end{equation}
Combining \eqref{eq:intermediate-FKG-SBP},\eqref{eq:intermediate-FKG-SBP-1} and \eqref{eq:intermediate-FKG-SBP-2} and recalling $S_\mathcal{U}\geq p_{N,\mathcal{U}}-o_N(1)$ and \eqref{eq:partition-sum-SBP}, \eqref{eq:def-Eiu}, we obtain
\begin{equation}
    1-p_{N,\mathcal{U}}\geq \frac{2p^2_{N,\mathcal{U}}}{9}-o_N(1).
\end{equation}
Solving this inequality yields $p_{N,\mathcal{U}}\leq \frac{3\sqrt{17}-9}{4}+o_N(1)$. The proof is complete.
\end{proof}

\section{Discussion}\label{sec:discuss}

We conclude this paper with some directions for future work. 
Strengthening the algorithmic success probability ruled out by Theorem~\ref{thm:stable-v1} is to us the most natural and interesting open problem.
\begin{problem}
    Show that stable algorithms cannot locate an isolated solution (in the sense of Theorem~\ref{thm:stable-v1}) with probability more than $o_N(1)$.
\end{problem}
We expect that achieving this will require new ideas. 
Indeed, a natural intuitive picture in the setting of Theorem~\ref{thm:stable-v1} would predict that after rerandomizing $\bG$ slightly, the number of solutions near the algorithm output is approximately a Poisson variable (since each point is unlikely to be a solution). A Poisson variable of course cannot concentrate on $1$. However, it can certainly place constant mass on $1$ (in fact as much as $1/e$), so such a heuristic will not bound the success probability by $o_N(1)$.

The next problem pertains to Theorem~\ref{thm:stable-SBP}.
\begin{problem}
    Show that in the generalized perceptron model considered in Theorem~\ref{thm:stable-SBP}, stable algorithms cannot $k_N$-locate a $k_N$-isolated solution with probability more than $1-\Omega(1)$, for some $k_N = \tilde\Omega(N)$.
\end{problem}
Combined with Corollary~\ref{cor:ld-stable}, such a result would imply that degree $\tilde O(N)$ algorithms cannot $k_N$-locate a $k_N$-isolated solution.
The assumption $k_N = \tilde O(\sqrt{N})$ is required in the proof of Theorem~\ref{thm:stable-SBP} to ensure that, when the perceptron constraints are defined by a finite union of intervals, at most one boundary constraint is active for all solutions in a ball of radius $k_N$; see \eqref{eq:condition-bG-eff-side}.
This step is the main obstacle to improving the bound to $k_N = \tilde\Omega(N)$.

Finally, we ask if a similar hardness result can be proven for the spherical perceptron.
Here one should interpret ``isolated solution'' to mean a solution cluster of small diameter, which is separated from the rest of the solution space by a significantly larger distance.

\begin{problem}
    Determine whether stable algorithms are able to locate a (suitably defined) isolated solution in a spherical perceptron model.
\end{problem}

\section*{Acknowledgment}
We thank Will Perkins, David Gamarnik and Yatin Dandi for bringing this problem to our attention and for helpful initial discussions. S. Gong is partially supported by National Key R\&D program of China (Project No. 2023YFA1010103) and the NSFC Key Program (Project No. 12231002).
B. Huang is supported by the Stanford Science Fellowship and the NSF Mathematical Sciences Postdoctoral Research Fellowship. 

 \footnotesize
 \bibliographystyle{alphaabbr}
\bibliography{bib}
\end{document}